\documentclass[lettersize,journal]{IEEEtran}
\usepackage{amsmath,amsfonts}
\usepackage{algorithmic}
\usepackage{array}
\usepackage{textcomp}
\usepackage{stfloats}
\usepackage{url}
\usepackage{verbatim}
\usepackage{graphicx}
\usepackage{balance}
\usepackage{xcolor}
\usepackage{subfig}
\def\BibTeX{{\rm B\kern-.05em{\sc i\kern-.025em b}\kern-.08em
		T\kern-.1667em\lower.7ex\hbox{E}\kern-.125emX}}
\usepackage{orcidlink}
\usepackage{makecell}
\usepackage{multirow}
\usepackage{booktabs} % provide \toprule, \midrule, bottomrule
\usepackage{cite}
\usepackage{amsmath,amssymb,amsfonts,amsthm}
\usepackage{algorithm}
\usepackage{flushend}
\newtheorem{lemma}{Lemma}
\newtheorem{thm}{Theorem}
\newtheorem{ass}{Assumption}
\newtheorem{remark}{Remark}

\begin{document}
	\title{Adaptive Event-triggered Formation Control of Autonomous Vehicles}
	
	\author{Ziming Wang\orcidlink{0000-0001-7000-9578} \IEEEmembership{Student Member,~IEEE}, Yihuai Zhang\orcidlink{0000-0002-7363-7796} \IEEEmembership{Student Member,~IEEE}, Chenguang Zhao\orcidlink{0000-0001-9314-908X} \IEEEmembership{Student Member,~IEEE} and Huan Yu$^{*}$\orcidlink{0000-0002-9324-0200} \IEEEmembership{Senior Member,~IEEE}

        \thanks{*Corresponding author}
		\thanks{Ziming Wang,  Yihuai Zhang, Chenguang Zhao and Huan Yu are with The Hong Kong University of Science and Technology (Guangzhou), Nansha, Guangzhou, 511458, Guangdong, China. (e-mail: huanyu@ust.hk)}
	}
	\maketitle
	
	\begin{abstract}
		This paper presents adaptive event-triggered formation control strategies for autonomous vehicles (AVs) subject to longitudinal and lateral motion uncertainties. The proposed framework explores various vehicular formations to enable safe and efficient navigation in complex traffic scenarios, such as narrow passages, collaborative obstacle avoidance, and adaptation to cut-in maneuvers. In contrast to conventional platoon control strategies that rely on predefined communication topologies and continuous state transmission, our approach employs a sampling-based observer to reconstruct vehicle dynamics. Building upon an adaptive backstepping continuous-time controller, we design three distinct event-triggered mechanisms, each offering a different trade-off between formation tracking performance and control efficiency by reducing the frequency of control signal updates. A Lyapunov-based stability analysis is conducted to guarantee bounded tracking errors and to avoid Zeno behavior. Finally, the proposed event-triggered controllers are validated through simulations of vehicular formation in three scenarios, highlighting their impact on traffic safety and mobility.
	\end{abstract}
	

	\section{Introduction}
	    For decades, the persistent challenges of traffic congestion and environmental degradation have driven intensive research into intelligent transportation systems~\cite{intro1,intro2,intro3,yihuai}. Among these initiatives, vehicle platoon control emerges as a promising solution to several pressing challenges in contemporary road transportation, given its capacity to markedly enhance highway capacity \cite{intro1}, improve safety \cite{intro2}, and reduce fuel consumption \cite{intro3}. The earliest implementation of vehicle platoon control can date back to the mid-eighties of the last century in the PATH program in California \cite{intro4,intro5}. The platoon is said to be homogeneous \cite{intro8.4}, if all vehicles have identical dynamics, otherwise it is called heterogeneous \cite{intro8.5}. The primary goal of platoon control is to ensure that achieving uniform velocity across all vehicles within the platoon formation while maintaining a predetermined inter-vehicle spacing policy. It has recently garnered significant research interest~\cite{intro6,intro7,intro8}.

        However, platoon control is limited to longitudinal regulation within fixed single-lane formation, focusing on uniform velocity and spacing. In contrast, formation control integrates both longitudinal and lateral dynamics, enabling adaptive responses to obstacles, traffic cut-ins, and multi-lane scenarios. This flexibility allows formation control to effectively navigate complex and variable traffic conditions, enhancing overall safety and efficiency.  Hence, we focus on a homogeneous vehicular formation control for AVs in linear, square and linear-queue scenarios in this paper.

        In most vehicular formation control research~\cite{intro3,intro8.1,intro9,intro10}, the majority of these communication connectivity models are based on algebraic graph theory, in which the nodes represent AVs and the edges represent the communication links between them. However, \cite{mu} highlights the limitations of using communication topologies. Despite the widespread implementation of vehicle-to-vehicle and vehicle-to-infrastructure technologies, disruptions in communication between certain AVs occur, leading to difficulties in maintaining stable tracking of the expected trajectory for the AVs formation. Moreover, it becomes challenging for a newly formed AVs formation, with additional AVs incorporated into the original formation, to maintain stable tracking of the initially expected trajectory. This difficulty arises from changes in the communication connectivity topology of the initial vehicular formation, rendering the Routh-Hurwitz stability criterion based on the original topology unattainable. This similar situation is also reflected in \cite{intro11,intro12}. Consequently, it is also necessary to consider the vehicular formation control problem that lacks information exchange. This paper proposes an adaptive formation control algorithm for AVs with motion uncertainty, utilizing sampling-based observers instead of relying on the communication connectivity topology. If the leading vehicle encounters communication delays or noise, only the final tracking trajectory of the entire AVs formation is affected.

        For sampling-based observer design, the first method in \cite{ob1} involves designing a sampling-based observer using a consistent approximation of the exact discretized model. Then \cite{ob2} focuses on creating a sampling-based observer based on a continuous observer of the continuous nonlinear model. Based on \cite{ob3}, the final approach is to design a continuous-sampling-based observer for a continuous nonlinear system, where the dynamics are continuous but the output is available only at specific sampling time points. Here in this paper, in most AVs, radar sensors and lidar are the primary method for observation, but measurement inaccuracies in real-world applications are unavoidable. Taking into account the communication limitations between vehicles mentioned before, we propose a sampling-observed adaptive formation control algorithm to tackle these issues.

        In the early stages of automatic control research \cite{eve0,backstepping}, the continuous-in-time control strategy was predominantly employed. This strategy often led to the unnecessary consumption of system resources. To tackle this problem, the event-triggered control (ETC) strategy has been developed. When a certain performance metric exceeds a predefined threshold, the ETC algorithm updates the control law. The debate over the superiority of one strategy over the other remains a subject of scholarly discourse until a seminal study conducted in 1999 \cite{eve1} demonstrated the advantages of event-triggered impulse control over periodic impulse control for a class of first-order stochastic systems, sparking widespread adoption and extension of ETC to both linear and nonlinear systems~\cite{eve2,eve3,chufa}. ETC offers an efficient alternative by updating control signals aperiodically, i.e., only when triggering conditions are met. This not only conserves system resources but also provides a rigorous framework for implementing continuous-in-time controllers on digital platforms.

        Based on the above development of ETC, over the past five years, many recent research efforts have developed distributed ETC strategies for vehicular formation control, addressing challenges posed by nonlinear dynamic models and resource-constrained onboard storage. \cite{eve4} proposed an ETC scheme that can reduce the control frequency and computation burden of the vehicles. Similarly, \cite{eve5} presented a periodic event-triggered fault detection filter to generate a residual signal for this intelligent and connected vehicles to actuator faults and external disturbances. However, their approaches primarily focus on periodic adjustments of triggering conditions, which exhibit limited effectiveness in reducing resource utilization and may inadvertently compromise control performance. Moreover, \cite{mu,eve6,eve7,eve8} introduced innovative triggering mechanisms to improve control efficiency in coordinating vehicular formation. Notably, the aforementioned studies mainly focus on optimizing resources, often neglecting whether the decreased actuation frequency in ETC strategies could potentially impact the safety and mobility of the entire vehicular formation. Therefore, this paper aims to address this aspect. By comparing the inter-vehicle spacing and time headway during the operation of vehicular formations under different control strategies, including the continuous-in-time controller and various ETC-based controllers, we will explore whether there are safety risks and factors affecting mobility in vehicular formations under different ETC strategies.
        
        Inspired by the above information, we propose a class of event-triggered adaptive formation control frameworks for AVs, taking into account motion uncertainties during the control process. Our approach utilizes the backstepping~\cite{backstepping}, neural networks (NNs)~\cite{NNs} and Lyapunov-based stability~\cite{chen}. Furthermore, the comparison among different vehicular formation control research is shown in Table \ref{comp}. The contributions of this paper mainly lie in three parts:
 
    \begin{enumerate}
        \item We propose an adaptive formation control framework that employs a sampling-based observer to address motion uncertainties. Through Lyapunov-based stability analysis, we ensure bounded tracking errors and robustness, while the observer reduces the need for continuous state transmission. Unlike previous studies that concentrate only on longitudinal control \cite{intro9,intro10,intro11,intro12}, our framework comprehensively addresses both longitudinal and lateral dynamics, enhancing its applicability to real-world scenarios.
        \item The proposed framework introduces a class of multi-threshold event-triggered mechanisms with different trigger conditions to further reduce the update frequency of the AVs' controllers, enhancing the scalability and stability of the vehicular formation system. Furthermore, in comparison to \cite{chufa,eve4,eve5,eve6,eve7}, this paper explores how various ETC strategies affect the traffic performance, including the safety and mobility of vehicular formations.
        \item Compared to \cite{mu}, this paper broadens the application scope by considering various vehicular formation scenarios: linear formation for navigating narrow passages, square formation for collaborative obstacle avoidance, and linear queue formation for accommodating cut-ins, thereby enhancing safety, stability, and operational versatility in diverse environments.
    \end{enumerate}

        The rest of the paper is organized as follows. Section~\ref{sec2} introduces the problem formulation, including the establishment of system model and sampling-based observer. The controller design and the multi-threshold event-triggered mechanism are presented in Section~\ref{sec3}. Section~\ref{sec4} presents the main theorem and the stability analysis procedure. In Section~\ref{sec5}, an illustrative example for vehicular formation demonstrates the effectiveness of the proposed adaptive event-triggered control algorithm under different scenarios. Finally, Section~\ref{sec6} draws the conclusion.

        \begin{table*}[!t]
		\caption{Comparison among different vehicular formation control research}
        \label{comp}
		\centering
		\renewcommand\arraybackslash{2}
		\begin{tabular}{ccccccc}
			\toprule[2pt] 
            \textbf{ } & \textbf{ } &
			\multicolumn{3}{c}{\textbf{Event-triggered control strategies}} & 
			\multicolumn{2}{c}{\textbf{Traffic performance}} \\
            \cmidrule(lr){3-5} \cmidrule(lr){6-7} 
			\textbf{Reference} & \textbf{Control direction} & \textbf{Fixed-threshold} & \textbf{Relative-threshold} & \textbf{Switched-threshold} & \textbf{Safety} & \textbf{Mobility} \\
			\midrule 
                Dolk et al., 2017\cite{db1}    & Longitudinal   & \checkmark   \\
                Wen et al., 2018\cite{db4}    & Longitudinal   & \text{ } & \text{ }  & \checkmark \\
                Keijzer and Ferrari 2019\cite{db3} & Longitudinal   & \checkmark   & \text{ } & \checkmark  \\
                Zhang et al., 2020\cite{db2}   & Longitudinal \& Lateral & \checkmark & \checkmark & \text{ } & \checkmark & \checkmark \\
                Ge et al., 2021\cite{intro6}   & Longitudinal  & \text{ }  & \checkmark & \text{ }  & \checkmark \\
                Wu et al., 2022\cite{eve7}     & Longitudinal   & \checkmark  & \checkmark  & \text{ }   & \checkmark\\
                Liu et al., 2023\cite{eve4}     & Longitudinal   & \checkmark   \\
			Xue et al., 2024\cite{mu}       & Longitudinal \& Lateral & \checkmark  \\
                Silva et al., 2025\cite{tian} & Longitudinal & \text{ } & \text{ }  & \checkmark & \checkmark \\
                This paper                      & Longitudinal \& Lateral & \checkmark & \checkmark  & \checkmark & \checkmark & \checkmark \\
			\bottomrule[2pt] 
		\end{tabular}
	\end{table*}

	\section{Problem Formulation}\label{sec2}
	In this section, we consider the vehicular formation tracking control problem with a class of threshold-based event-triggered control strategies, the designed sampling-based observer is used to convert the sampled vehicle positions into vehicle speeds and positions. This control algorithm is used in different scenarios. Fig.\ref{fig1} shows three common scenarios in traffic systems, details are as follows:

    \begin{figure}
            \centering
            \subfloat[Linear formation: navigating narrow passages.]{\includegraphics[width=1\linewidth]{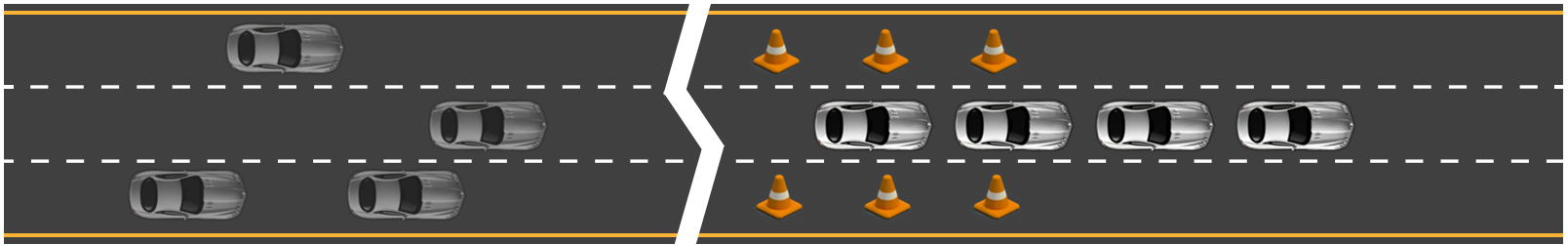}}\\
            \vspace{-8pt}
            \subfloat[Square formation: collaborative obstacle maneuvering.]{\includegraphics[width=1\linewidth]{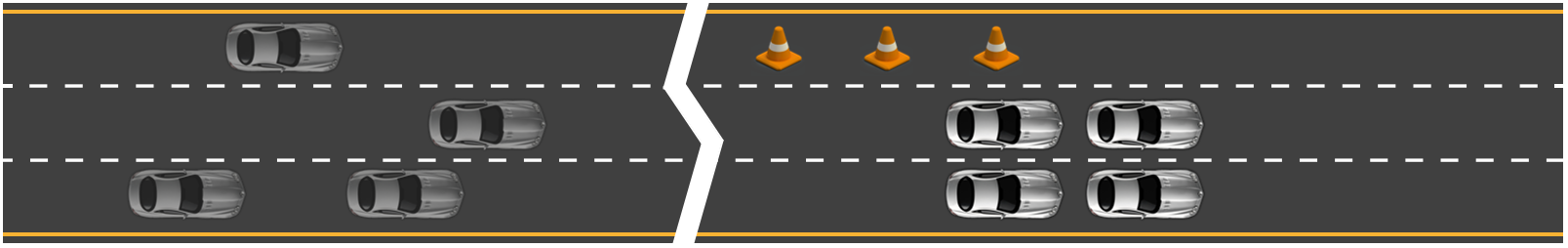}}\\
            \vspace{-8pt}
            \subfloat[Linear-queue formation: flexible spacing for cut-ins.]{\includegraphics[width=1\linewidth]{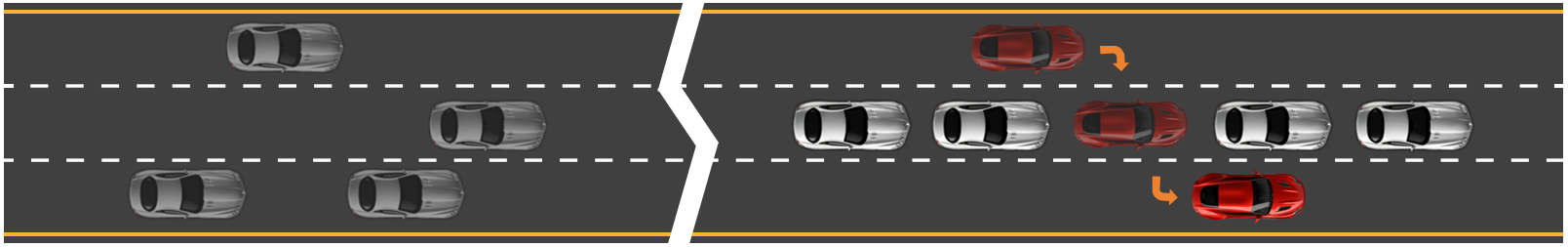}}
            \caption{AVs formation control in different traffic scenarios.}
            \label{fig1}
        \end{figure}
        
        \begin{enumerate}
            \item Linear formation: adopted for navigating narrow passages or lane reductions, ensuring safe and efficient traversal through a single lane.
            \item Square formation: utilized for rapid obstacle avoidance and enhanced spatial awareness, enabling AVs to collaboratively maneuver around static hazards.
            \item  Linear-queue formation: designed to accommodate cut-in vehicles by enabling flexible reconfiguration of inter-vehicle spacing, thereby maintaining safety and stability.
        \end{enumerate}
	
	\subsection{Dynamical Modeling of AVs}
	In this paper, in terms of setting the distance between AVs, we fully consider the relationship between the length of the vehicle itself and the spacing at the front and rear of the vehicle. Therefore, when studying the problem of adaptive formation control, we consider the AVs' formation as a mass, resulting in a multi-agent system. We consider that there are $N$ AVs with $i=1,\cdots, N$ being the index of each vehicle. We define $m_i$ as the mass of AV-$i$ and $D_i$ as the drag factor. We formulate the system model as:
	\begin{equation}
		\label{equ1}
		\begin{aligned}
			\dot{x}_i&=v_i, \\
			\dot{v}_i&=u_i+D_i,
		\end{aligned}
	\end{equation}
	where $x_i=[x_i^x,x_i^y]^T$, with $x_i^x$ and $x_i^y$ being the longitudinal and lateral positions of AV-$i$; $v_i=[v_i^x,v_i^y]^T$ represent the longitudinal and lateral speeds of AV-$i$; $u_i=[\tau_i^x/m_i,\tau_i^y/m_i]^T$ is the control variable with the longitudinal and lateral engine traction force $\tau_i^x$, $\tau_i^y$. The drag factor is $D_i = F_i + P_i$, with $F_i=[f_i^x/m_i,f_i^y/m_i]^T$ and $P_i=[\rho_i^x/m_i,\rho_i^y/m_i]^T$, where $f_i^x$ and $f_i^y$ are the unknown longitudinal and lateral resistances, and $\rho_i^x$ and $\rho_i^y$ are the unknown longitudinal and lateral disturbances.  %Let $y_i=x_i$, where $i\in N$, $y_i\in R^2$ denotes the output of the system. 
    We take AV-$1$ as the leading vehicle, and other AVs as following vehicles. Each AV-$i$ follows AV-$i-1$ for $i=2,\cdots,N$.
    We design the controller to ensure that each following vehicle efficiently and synchronously tracks the trajectory and speed of its preceding vehicle. 
    %In addition, $||\cdot||$ denotes the Euclidean norm, $|\cdot|$ denotes the absolute value function.

	\subsection{Observer Design}
    In adaptive control, to approximate nonlinear functions and handle uncertainties in dynamic systems, we introduce radial basis function neural networks, enabling the system to effectively adapt to changes and maintain optimal performance under uncertain conditions.  For the unknown nonlinear function $D_i$ in the system (\ref{equ1}), we assume that there exists a basis function $\Lambda_i$ and a bounded approximation error $\sigma_i$ such that $D_i={W_i^*}^T\Lambda_i(\gamma_i)+\sigma_i$, with $\Lambda_i(\gamma_i)=[\Lambda_{i,1}^T(\gamma_i),\Lambda_{i,2}^T(\gamma_i)]\in \mathbb{R}^{2l\times1}$, $\gamma_i$ represents the network's input, has its dimensionality defined by the number of input variables. Here, $\Lambda_i(\gamma_i)=\exp[-(\gamma_i-\gamma_i^*)^T(\gamma_i-\gamma_i^*)/\varphi_i^2]$, where $\gamma_i^*$ stands for the center of the receptive field, $\varphi_i$ is known as the width of Gaussian function. $W_i^*=[W_{i,1}^{*T}, o_{1\times l}; o_{1\times l}, W_{i,2}^{*T}]^T\in \mathbb{R}^{2l\times2}$ stands for the optimal weight vectors with the null vector $o\in \mathbb{R}^{1\times l}$ and the count of tunable weights in the hidden layer of the neural networks is $l$, with $W_{i,j}^{*}\in \mathbb{R}^{l\times1}$, $j=1,2$. $\hat{W}_i$ is the estimation matrix of $W_i^*$, with $\widetilde{W}_i=\hat{W}_i-W_i^*$ being the estimation error. Then, with $i=1,\cdots, N$, the optimal weight $W_i^*$ is defined as
	\begin{equation}
		\label{equ2}
		\begin{aligned}
			W_i^*= {\mathop{\arg\min}}_{W_i}\{{\mathop{\sup}}_{\gamma_i}||W_i^T\Lambda_i(\gamma_i)-D_i||\}.
		\end{aligned}
	\end{equation}

In this paper, we consider the inaccuracies of the sampling information caused by radar and lidar, which are used by most AVs at present. In order to eliminate the impact of sampling information errors in the formation control system, we introduce a sampling-based observer $\hat{x}_i$ that estimates the actual position and speed information by observing the inaccurate position data $\overline{x}_i$, which represents the sampled value at the sampling time $t_k$. Then we design the sampling-based observer as
	\begin{equation}
		\label{equ3}
		\begin{aligned}
			\dot{\hat{x}}_i&=\hat{v}_i+C_{i,1}(\overline{x}_i-\hat{x}_i), \\
			\dot{\hat{v}}_i&=u_i+C_{i,2}(\overline{x}_i-\hat{x}_i)+\hat{W}_i^T\Lambda_i(\gamma_i),
		\end{aligned}
	\end{equation}
	where $C_{i,1}\in \mathbb{R}^{2\times2}$ and $C_{i,2}\in \mathbb{R}^{2\times2}$ are designed output injection gains.
	
	\section{Main Result}\label{sec3}

    	\begin{table}[!tbp]
		\caption{Notations}%title
        \label{notation}
		\centering
		\renewcommand\arraybackslash{2}
		\begin{tabular}{cc}% four columns
			\toprule[2pt] %change the first line to \toprule
			$x_i$, $v_i$ & the position and speed of AV-$i$ \\
			\midrule %change the second line to midrule
			$x_i^r$, $v_i^r$ & the expected position and speed of AV-$i$  \\
			\midrule
			$\hat{x}_i$, $\hat{v}_i$ & the observed position and speed of AV-$i$ \\
			\midrule
			$z_{i,1}$, $z_{i,2}$ & the position and speed tracking errors of AV-$i$  \\
			\midrule
                $e_{i,1}$, $e_{i,2}$ & the position and speed observation errors of AV-$i$  \\
			\midrule
			$l_i$ & the expected inter-vehicle distances of AV-$i$ \\
                \midrule
			$\alpha_i$ & the backstepping based virtual controller of AV-$i$ \\
			\midrule
			$\mu_i$ & the continuous-in-time controller of AV-$i$ \\
			\midrule
			$w_i$ & the ETC based updated controller of AV-$i$ \\
			\bottomrule[2pt] %change the third line to bottomrule
		\end{tabular}
	\end{table}
    
    We consider the AVs formation control problem for notations in Table \ref{notation}. In this research, the first leading vehicle leads the total vehicular formation to track the expected position $x_1^r=[x_1^{rx},x_1^{ry}]^T$ and speed $v_1^r=[v_1^{rx},v_1^{ry}]^T$. The remaining vehicles' expected positions $x_i^r=[x_i^{rx},x_i^{ry}]^T$ and speeds $v_i^r=[v_i^{rx},v_i^{ry}]^T$, $i=2,\cdots, N$. Notably, $\hat{x}_{i-1}=[\hat{x}_{i-1}^x,\hat{x}_{i-1}^y]^T$, $i=2,\cdots, N$, contains the observed single preceding vehicle's longitudinal and lateral position, and $l_i=[l_i^x,l_i^y]^T$, $i=1,\cdots, N$, contains the expected longitudinal and lateral inter-vehicle distances. Then we obtain
	\begin{equation}
		\label{equ4}
		\begin{aligned}
			x_i^{rx}&=\hat{x}_{i-1}^x-l_i^x,\\
			x_i^{ry}&=\hat{x}_{i-1}^y-l_i^y,\\
			\dot{x}_i^{rx}&=v_{i-1}^{rx}, \\
			\dot{x}_i^{ry}&=v_{i-1}^{ry}.
		\end{aligned}
	\end{equation} 
	
	With a designed positive definite matrix $K_{i,1}\in \mathbb{R}^{2\times2}$, we define the position tracking error and speed tracking error for position $z_{i,1}$ and speed $z_{i,2}$, and based on the backstepping method, the virtual controller $\alpha_i$ as
	\begin{equation}
		\label{equ5}
		\begin{aligned}
			z_{i,1}&=\hat{x}_i-x_i^r,\\
			z_{i,2}&=\hat{v}_i-\dot{x}_i^r-\alpha_i,\\
            \alpha_i &=-K_{i,1}z_{i,1}.
		\end{aligned}
	\end{equation}
	
	We utilize the Lyapunov function composed of these errors and neural network factors to analyze the control performance. If its derivative at the equilibrium point is zero, positive elsewhere, and its time derivative is negative, this indicates that the errors are controlled within a certain range and that the system is asymptotically stable. Thereafter, define the Lyapunov function as
	\begin{equation}
		\label{equ6}
		\begin{aligned}
			V_0=&e^TPe,\\
			V_i= &\frac{1}{2}z_{i,1}^Tz_{i,1}+\frac{1}{2}z_{i,2}^Tz_{i,2}+\frac{1}{2}\widetilde{\sigma}_i^T\Delta_i\widetilde{\sigma}_i  \\
			&+\frac{1}{2}\widetilde{W}_{i,1}^TO_{i,1}^{-1}\widetilde{W}_{i,1}+\frac{1}{2}\widetilde{W}_{i,2}^TO_{i,2}^{-1}\widetilde{W}_{i,2},
		\end{aligned}
	\end{equation}
	where $O_{i,1},O_{i,2}\in \mathbb{R}^{l\times l}$, $\Delta_i\in \mathbb{R}^{2\times2}$ are designed constant positive definite diagonal matrices. $P \in \mathbb{R}^{4n\times4n}$ denotes a positive symmetric matrix,  $e=[e^1,e^2]^T \in \mathbb{R}^{4n\times1}$, where $e^1=[e_{1,1},e_{2,1},...,e_{n,1}]^T$, $e^2=[e_{1,2},e_{2,2},...,e_{n,2}]^T$, with $e_{i,1}$, $e_{i,2}$ representing the error in the observed and actual position and speed respectively, i.e., $e_{i,1}=x_i-\hat{x}_i$ and $e_{i,2}=v_i-\hat{v}_i$. 
	
	According to \cite{back}, to ensure the robustness and adaptability of the system and to enhance the overall performance and stability, we introduce a recursive technique to design adaptive controllers, the backstepping method. The system is decomposed into multiple subsystems, each equipped with a virtual controller. The design process begins with the simplest subsystem and gradually advances to design more complex controllers. At each step, the continuous-in-time controller $\mu_i$ and parameter update laws $\hat{W}_{i,j}$, $\hat{\sigma}_i$ are designed as
	\begin{equation}
		\label{equ7}
		\begin{aligned}
			\mu_i&=-K_{i,2}z_{i,2}-z_{i,1}-\hat{W}_i^T\Lambda_{i}(\gamma_i)-\kappa(z_{i,2})\hat{\sigma}_i+\dot{\alpha}_i+\ddot{x}^r_i,\\
			\dot{\hat{W}}_{i,j}&=O_{i,j}(\Lambda_{i,j}z_{i,2,j}-\Xi_{i,j}\hat{W}_{i,j}),\\
			\dot{\hat{\sigma}}_i&=\Delta_i[\kappa_i(z_{i,2})z_{i,2}-\Upsilon_i(\hat{\sigma}_i-\sigma^0_i)],
		\end{aligned}
	\end{equation}
	where $i\in N$, $K_{i,2}$ is a constant positive definite diagonal matrix, $\Xi_{i,j}\in \mathbb{R}$ is a positive constant, $\Upsilon_i\in \mathbb{R}^{2\times2}$ is a constant positive definite diagonal matrix, $\kappa(z_{i,2})={\rm diag}\{{\rm sgn}(z_{i,2,1}),{\rm sgn}(z_{i,2,2})\}\in \mathbb{R}^{2\times2}$ where ${\rm sgn}(\cdot)$ is the signum function. $z_{i,2,j}$ is the $j$th component of vector $z_{i,2}$ where $j=1,2$.

    \begin{remark}
  To model the vehicular formation tracking problem for AVs, we introduce the system (\ref{equ1}), including dynamical drag factor $D_i$, which is approximated by the radial basis function neural networks in (\ref{equ2}). We design the sampling-based observer for state estimation from inaccurate sampled data collected by radar and lidar in (\ref{equ3}). We set the desired trajectory for the vehicular formation in (\ref{equ4}), and define the position and speed tracking errors for each vehicle in (\ref{equ5}). The Lyapunov function is defined in (\ref{equ6}) is intended to facilitate the analysis of control performance. Before the stability analysis, based on the backstepping method, we design the continuous-in-time controller and parameter update laws in (\ref{equ7}).
\end{remark}
	
	For practical implementation of AVs formation control ($\ref{equ7}$), it is unrealistic to have continuous transmission of control signals. To overcome this problem, the event-triggered control method is employed for digital implementation. In what follows, we present three different threshold-based event-triggered control strategies to demonstrate the advantages of different triggering conditions in optimizing the controller, facilitating subsequent scholars in discerning the selection of event-triggered strategies for AVs formation problems. The control procedure is concluded in Algorithm \ref{algorithm1}.
	
	\begin{algorithm}[!t]
		\caption{Algorithm of the event-triggered control}
		\label{algorithm1}
		\renewcommand{\algorithmicrequire}{\textbf{Input:}}
		\renewcommand{\algorithmicensure}{\textbf{Output:}}
		
		\begin{algorithmic}[1]
			\REQUIRE Basic parameters: N, T, t

                ETC parameters: $\varsigma_i$, $\overline{\varsigma}_i$, $\varepsilon_{i,1}$, $\varepsilon_{i,2}$, $\zeta_i$, $\xi_i$, $\overline{\xi}_i$, $S$
            
			ETC strategies switching parameter: $Case$

			\ENSURE $u_i$
			\FOR{$i$ in [$N$]}
			\STATE Calculate the tracking error $z_{i,1}$, $z_{i,2}$, the original controller $\mu_i$ and the adaptive parameter $W_{i,j}$ by Eq. (\ref{equ5}) and Eq. (\ref{equ7})
			\FOR{$t$ in [$T$]}
			\STATE \# Trigger condition unmet, controller unchanged
			\STATE $u_i(t)=u_i(t)$  
			\IF {$Case=1$}
			\IF {$|e_i|\geq\varsigma_i$}
			\STATE Calculate $w_i(t)$ by Eq. (\ref{equ8})
			\STATE $u_i(t)=w_i(t_i^k)$
			\ENDIF
			\ELSIF{$Case=2$}
			\IF {$|e_i|\geq\zeta_i|u_i(t)|+\xi_i$}
			\STATE Calculate $w_i(t)$ by Eq. (\ref{equ10})
			\STATE $u_i(t)=w_i(t_i^k)$
			\ENDIF
			\ELSIF{$Case=3$}
			\IF {$|u_i|<S$ \& $|e_i|\geq\varsigma_i$}
			\STATE Calculate $w_i(t)$ by Eq. (\ref{equ8})
			\STATE $u_i(t)=w_i(t_i^k)$
			\ELSIF {$|u_i|\geq S$ \& $|e_i|\geq\zeta_i|u_i(t)|+\xi_i$}
			\STATE Calculate $w_i(t)$ by Eq. (\ref{equ10})
			\STATE $u_i(t)=w_i(t_i^k)$
			\ENDIF
			\ENDIF
			\ENDFOR
			\ENDFOR
			\RETURN $u_i$
		\end{algorithmic}
	\end{algorithm}

	\subsection{Fixed-threshold Strategy}
	This strategy utilizes a fixed threshold $\varsigma_i$ to activate controller updates, initiating control signal adjustments solely when the measurement error surpasses a predefined constant boundary. In this case, with the continuous-in-time controller $\mu_i$ in (\ref{equ7}), the updated controller in the fixed-threshold strategy is reformulated as	
	\begin{equation}
		\label{equ8}
		\begin{aligned}
			w_i(t)=\mu_i-\overline{\varsigma}_i
			\begin{bmatrix}
				\tanh(\frac{\overline{\varsigma}_iz_{i,2,1}}{\varepsilon_{i,1}})\\
				\tanh(\frac{\overline{\varsigma}_iz_{i,2,2}}{\varepsilon_{i,2}})
			\end{bmatrix}.
		\end{aligned}
	\end{equation}
        
        The triggering event is defined as
	\begin{equation}
		\label{equ9}
		\begin{aligned}
			u_i(t)&=w_i(t_i^k),\\
			t_i^{k+1}&=\inf\{t>t_i^k|\lvert\lvert e_i(t)\rvert\rvert\geq\varsigma_i\},
		\end{aligned}
	\end{equation}
	where $e_i(t)=w_i(t)-u_i(t)$ denotes the measurement error, $\varepsilon_{i,1}$, $\varepsilon_{i,2}$, $\overline{\varsigma}_i$ and $\varsigma_i$ are all designed positive constants, $\overline{\varsigma}_i>\varsigma_i$. $t_i^k$ is the controller update time, i.e. whenever ($\ref{equ9}$) is satisfied, the time will be marked as $t_i^{k+1}$ and the control value $u_i(t_i^{k+1})$ will be applied into the system. During the period $t_i^k\leq t_i<t_i^{k+1}$, the control signal holds as a constant, i.e. $w_i(t_i^k)$.

	\subsection{Relative-threshold Strategy}	
	In ($\ref{equ9}$), it is observed that the threshold $\varsigma_i$ remains constant regardless of the magnitude of the control signal. Nonetheless, in the context of stabilization challenges, akin to the concept proposed in \cite{relative}, it is advantageous to employ a dynamic threshold for the triggering event. Specifically, when the control signal $u_i$ is high, allowing for a larger measurement error is beneficial, as it results in extended update intervals. Conversely, as the system states converge towards the equilibrium point at zero, the control signal $u_i$ diminishes, necessitating a reduced threshold to enable more precise control, thereby enhancing system performance. Based on this consideration, we introduce the following relative-threshold control strategy. In this case, with the continuous-in-time controller $\mu_i$ in (\ref{equ7}), the updated controller in the relative-threshold strategy is reformulated as
	\begin{equation}
		\label{equ10}
		\begin{aligned}
			w_i(t)=-(1+\zeta_i)(\mu_i
			\begin{bmatrix}
				\tanh(\frac{\mu_iz_{i,2,1}}{\varepsilon_{i,1}})\\
				\tanh(\frac{\mu_iz_{i,2,2}}{\varepsilon_{i,2}})
			\end{bmatrix}
			+\overline{\xi}_i
			\begin{bmatrix}
				\tanh(\frac{\overline{\xi}_iz_{i,2,1}}{\varepsilon_{i,1}})\\
				\tanh(\frac{\overline{\xi}_iz_{i,2,2}}{\varepsilon_{i,2}})
			\end{bmatrix})
		\end{aligned}
	\end{equation}
	
	The triggering event is defined as
	\begin{equation}
		\label{equ11}
		\begin{aligned}
			u_i(t)&=w_i(t_i^k)\\
			t_i^{k+1}&=\inf\{t>t_i^k|\lvert\lvert e_i(t)\rvert\rvert\geq\zeta_i|u_i(t)|+\xi_i\}
		\end{aligned}
	\end{equation}
	where $e_i(t)=w_i(t)-u_i(t)$ denotes the measurement error, $\varepsilon_{i,1}$, $\varepsilon_{i,2}$, $\xi_i$, $0<\zeta_i<1$ and $\overline{\xi}_i>\xi_i/(1-\zeta_i)$ are all designed positive constants.
	
	\subsection{Switched-threshold Strategy}	
	In this part, we present a switched-threshold strategy. The relative-threshold strategy adapts the threshold in proportion to the magnitude of the control signal. This means that when the control signal is large, the threshold increases, allowing for longer update intervals, and as the control signal approaches zero, the control becomes more precise, enhanced system performance. Nevertheless, an excessively large control signal may cause substantial measurement errors, resulting in abrupt signal changes during controller updates. In contrast, the fixed-threshold strategy maintains a consistent upper bound on measurement errors, regardless of the control signal size. Taking these factors into account, we propose the switched-threshold strategy below. According to \cite{chufa}, with switching boundary $S$ and parameters $\zeta_i$, $\xi_i$, $\varsigma_i$, the triggering event is defined as
	\begin{equation}
		\label{equ14}
		\begin{aligned}
			u_i(t)&=w_i(t_i^k)\\
			t_{s+1}&=\left\{
			\begin{aligned}
				\inf\{t\in R||e_i(t)|\geq\zeta_i|u_i(t)|+\xi_i\} &, \mbox{if }|u_i(t)|< S \\
				\inf\{t\in R||e_i(t)|\geq\varsigma_i\} \quad\quad\quad\quad\quad&, \mbox{if }|u_i(t)|\geq S \\
			\end{aligned}
			\right.
		\end{aligned}
	\end{equation}
	
	The switched-threshold strategy merges elements from the initial two threshold strategies. It operates by utilizing the relative-threshold strategy when the magnitude of the control signal, $|u_i(t)|$, is less than a predefined value $S$, This allows for precise control when necessary. In contrast, if $|u_i(t)|$ exceeds this boundary, the system switches to the fixed-threshold method to ensure that the measurement error remains within a controlled range, thus safeguarding the system performance. Therefore, this strategy can not only obtain a reasonable update interval but also avoid the excessive large impulse.
	
\begin{remark}
  The three ETC strategies offer distinct trade-offs between formation tracking control precision and resource efficiency. The fixed-threshold strategy (\ref{equ8})(\ref{equ9}) initiates updates when the measurement error of the controllers surpasses a constant boundary $\varsigma_i$, ensuring periodic triggering but lacking precise control over the triggering condition. On the other hand, the relative-threshold strategy (\ref{equ10})(\ref{equ11}) adjusts the threshold $\zeta_i|u_i(t)|+\xi_i$ in proportion to the controller value, improving resource efficiency during significant maneuvers while enhancing precision near equilibrium. The switched-threshold strategy (\ref{equ14}) combines both strategies. It optimally balances computational load and formation tracking performance, reducing abrupt control changes while maintaining adaptive error constraints.
\end{remark}

	\section{Stability Analysis}\label{sec4}
	In this section, We rigorously evaluate these three types of ETC strategies to confirm the stability of the closed-loop system and its tracking performance. Furthermore, we demonstrate that the proposed strategies satisfy the requirement of avoiding Zeno behavior. These results are formally presented in the following theorems.

    \begin{thm}
        Based on the closed-loop system $(\ref{equ1})$, the continuous-in-time controller ($\ref{equ7}$) and the updated controller ($\ref{equ8}$) under the fixed-threshold strategy ($\ref{equ9}$), the overall AVs formation could track the expected position and speed trajectories, and all signals within the closed-loop system are bounded. In other words, the position and speed tracking errors $z_{i,1}$ and $z_{i,2}$ in ($\ref{equ5}$), the sampling-based observer error $e$ and the parameter estimation errors $\widetilde{W}_{i,j}$ and $\widetilde{\sigma}_i$ in ($\ref{equ6}$) are all bounded. Mathematically, with the positive upper bounds of the corresponding errors $B_x$, $B_v$, $B_{z_1}$, $B_{z_2}$, $B_W$ and $B_\sigma$, it holds that
            
        	\begin{equation}
        		\label{equ15}
        		\begin{aligned}
                        &|\hat{x}_i-x_i|\leq B_x\\
                        &|\hat{v}_i-v_i|\leq B_v\\
                        &|\hat{x}_i-x_i^r|\leq B_{z_1}\\
                        &|\hat{v}_i-\dot{x}_i^r-\alpha_i|\leq B_{z_2} \\
        			&||\hat{W}_{i,j}-W_{i,j}^*||\leq B_W\\
        			&|\hat{\sigma}_i-\sigma_i|\leq B_\sigma
        		\end{aligned}
        	\end{equation}
        
            Meanwhile, there exists a time $t^*$ such that the inter-execution interval $\{t_i^{k+1}-t_i^k\}$ are lower bounded by $t^*$, the Zeno behavior is effectively excluded.
    \end{thm}
 
  Prior to conducting the stability proof, the following lemmas and assumptions are indispensable.

    \begin{lemma}[\cite{lemma0,lemma1}]
        For any constants $\vartheta_1, \vartheta_2\in \mathbb{R}$, the following inequality $0\leq |\vartheta_1|-\vartheta_1\tanh(\frac{\vartheta_1}{\vartheta_2})\leq\varepsilon^*\vartheta_2$ holds, where $\varepsilon^*$ is a constant that satisfies $\varepsilon^*=e^{-(\varepsilon^*+1)}$.
    \end{lemma}
      \begin{ass}[\cite{mu}]
        The position function $x_i(t)$ of AV-$i$ is a differentiable and Lipschitz continuous function. There exists a positive constant $g_i\in R$ such that  $||x_i(t_k)-x_i(t)||\leq g_i(t-t_k)\leq g_{\max}(t-t_k)$ hold for all $i=1,\cdots, N$, with $g_{\max}=\mathop{\max}_{i\in N}\{g_i\}$. 
    \end{ass}
	\begin{ass}[\cite{mu}]
        We assume that the error between the sampled value and the actual value at the moment of sampling remains within a certain range $\varrho_i\in \mathbb{R}$ such that $||\overline{x}_i(t_k)-x_i(t_k)||\leq\varrho_i\leq\varrho_{\max}$ hold for all $i=1,\cdots, N$, with $\varrho_{\max}=\mathop{\max}_{i\in N}\{\varrho_i\}$. 
	\end{ass}

    \begin{proof}
         Consider about the fixed-threshold event-triggered strategy, from (\ref{equ9}), we have $|w_i(t)-u_i(t)|\leq\varsigma_i$. Therefore, there exists a time-varying parameter $\theta_i(t)\leq1$, satisfying $\theta_i(t_i^k)=0$ and $\theta_i(t_i^{k+1})=\pm1$, with $r_i=[1,1]^T$, one gets
	\begin{align}
		w_i(t)=\mu_i(t)+\theta_i\varsigma_ir_i\label{39}
	\end{align}

    Based on \emph{Lemma 1}, we define a constant $\varepsilon^*$, then we get
    \begin{align}
		&z_{i,2}^T(-\theta_i\varsigma_ir_i-\overline{\varsigma}_i\begin{bmatrix}
			\tanh(\frac{\overline{\varsigma}_iz_{i,2,1}}{\varepsilon_{i,1}})\\
			\tanh(\frac{\overline{\varsigma}_iz_{i,2,2}}{\varepsilon_{i,2}})
		\end{bmatrix})\leq\varepsilon^*(\varepsilon_{i,1}+\varepsilon_{i,2})
	\end{align}
	
	To address the convergence of all signals in the closed-loop system, we establish the Lyapunov function candidate for the entire system as follows:
	\begin{equation}
		\label{equ16}
		\begin{aligned}
			V=V_0+\sum\limits_{i=1}^{n} V_i.
		\end{aligned}
	\end{equation}

    Define $\lambda_{i,1}=C_{i,1}(\overline{x}_i-x_i)$ and $\lambda_{i,2}=C_{i,2}(\overline{x}_i-x_i)$. 
    Based on \emph{Assumption 1} and \emph{Assumption 2}, with unknown bound $\lambda_i^{\max}$ and $\varpi_i^{\max}$, we know that $\overline{x}_i(t_k)-x_i(t)$ is bounded, so $\lambda_{i,1}$ is bounded, with $||\lambda_{i,1}||\leq \lambda_i^{\max}\in \mathbb{R}$. Similarly, $\sigma_i-\lambda_{i,2}$ is also bounded, with $||\sigma_i-\lambda_{i,2}||\leq\varpi_i^{max}\in \mathbb{R}$. Then taking the derivative of $V_0$ and the Young's inequality, we obtain
    \begin{align}
		\dot{V}_0\leq&-\phi_{\min}(Q)||e||^2-3||P||^2||e||^2+\sum\limits_{i=1}^{n}\sum\limits_{j=1}^{2}\widetilde{W}_{i,j}^T\widetilde{W}_{i,j} \nonumber \\
		&+\sum\limits_{i=1}^{n}(\lambda_i^{\max})^2+\sum\limits_{i=1}^{n}(\varpi_i^{\max})^2, \label{45}
	\end{align}
	where $\phi_{\min}(Q)$ is the minimum eigenvalue of the positive definite matrix $Q$, and in this paper, $\phi_{\min}(\cdot)$ presents the minimum eigenvalue of the positive definite matrix $(\cdot)$.
	
	Taking the derivative of $V_i$ and the Young's inequality, with the unknown bound of the optimal weight vector $\overline{W}_{i,j}^*$, we obtain
	\begin{equation}
		\label{equ17}
		\begin{aligned}
			\dot{V}_i&\leq -K_{i,1}z_{i,1}^Tz_{i,1}-K_{i,2}z_{i,2}^Tz_{i,2}+\frac{1}{2}\sum_{j=1}^2\Xi_{i,j}(\overline{W}_{i,j}^*)^2\\&-\frac{1}{2}\sum_{j=1}^2 
			\Xi_{i,j}\phi_{\min}(O_{i,j})\widetilde{W}_{i,j}^TO_{i,j}^{-1}\widetilde{W}_{i,j}+\varepsilon^*(\varepsilon_{i,1}+\varepsilon_{i,2}) \\
			&-\frac{1}{2}\phi_{\min}(\Upsilon_i\Delta_i)\widetilde{\sigma}_i^T\Delta_i^{-1}\widetilde{\sigma}_i+\frac{1}{2}(\sigma_i-\sigma_i^0)^T\Upsilon_i(\sigma_i-\sigma_i^0).
		\end{aligned}
	\end{equation}

    \begin{figure*}
            \centering
            \subfloat{\includegraphics[width=0.333\linewidth]{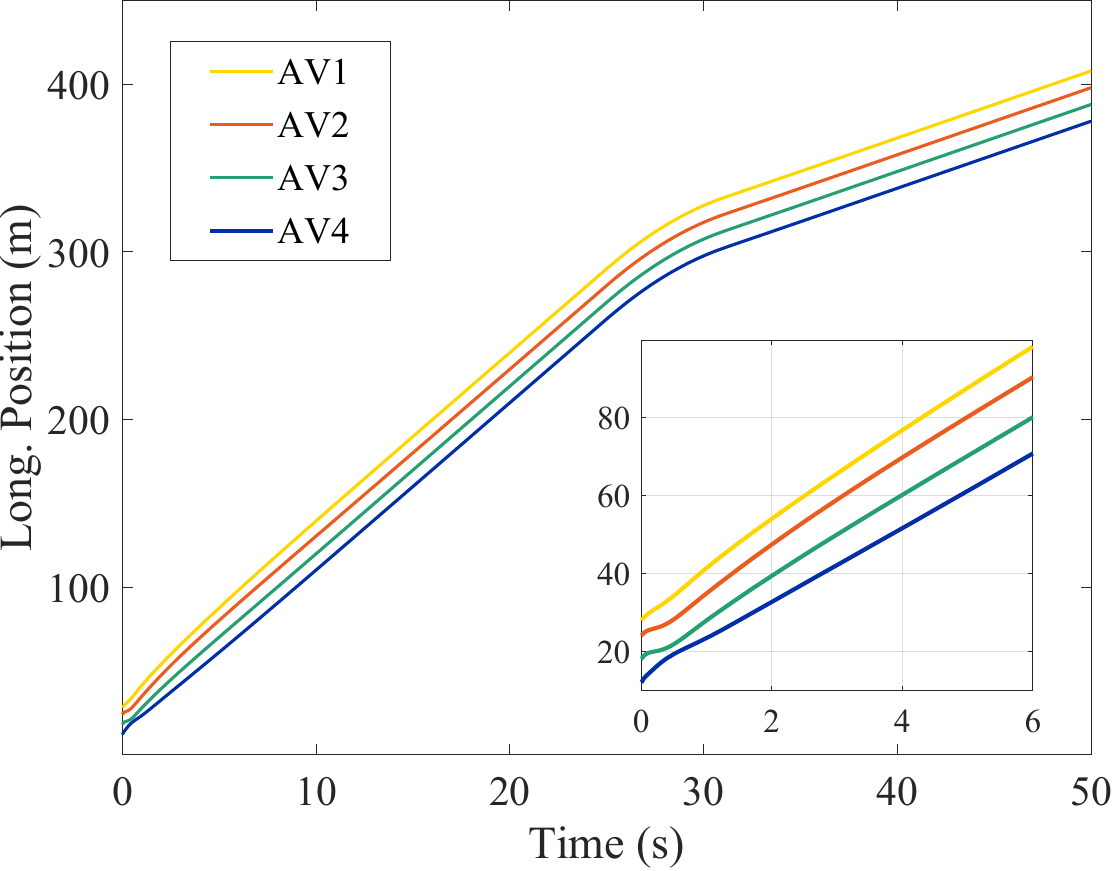}}
            \subfloat{\includegraphics[width=0.333\linewidth]{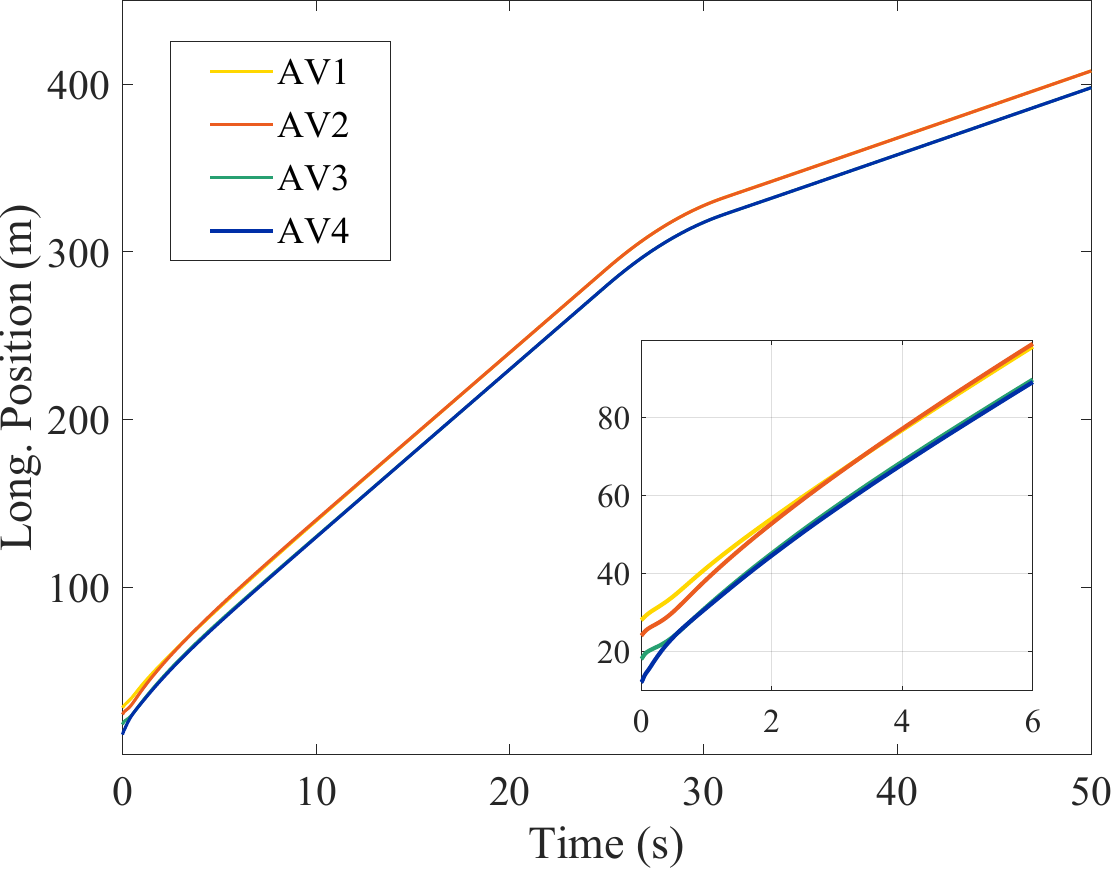}}
            \subfloat{\includegraphics[width=0.333\linewidth]{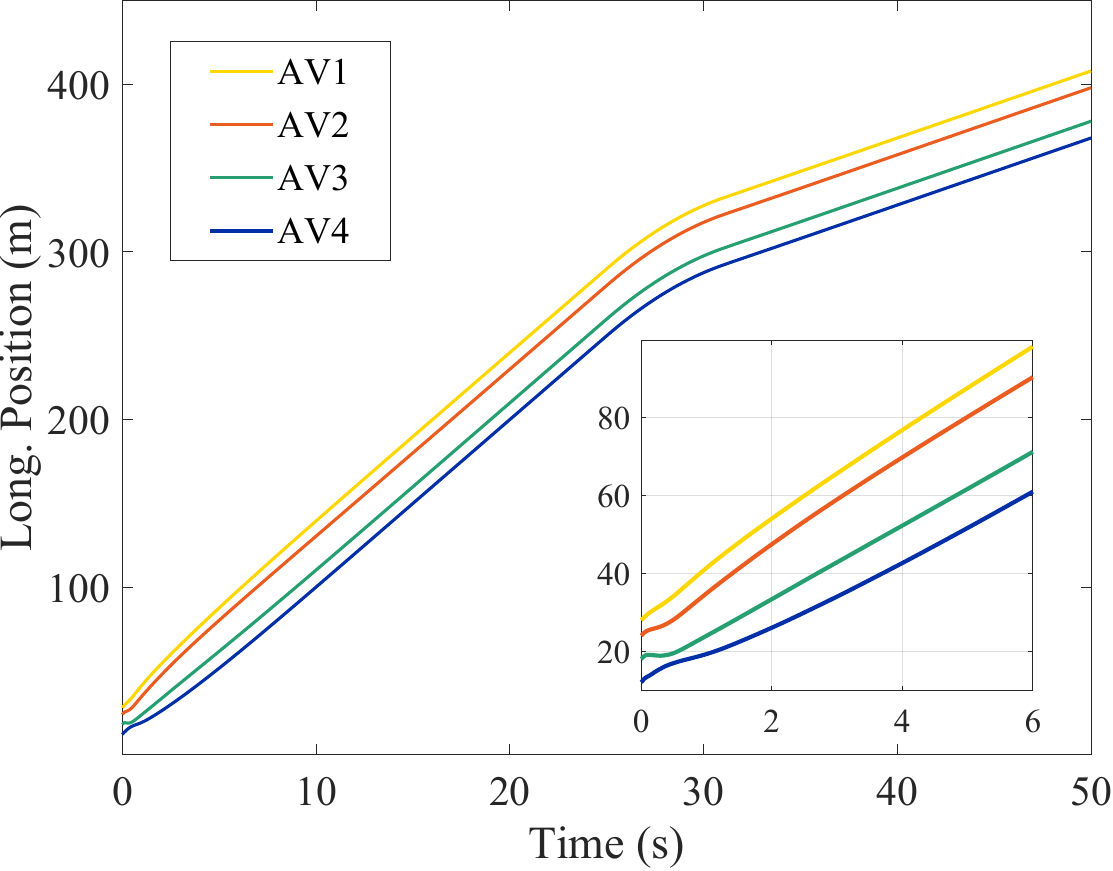}}\\
            \setcounter{subfigure}{0}
            \subfloat[Linear formation]{\includegraphics[width=0.33\linewidth]{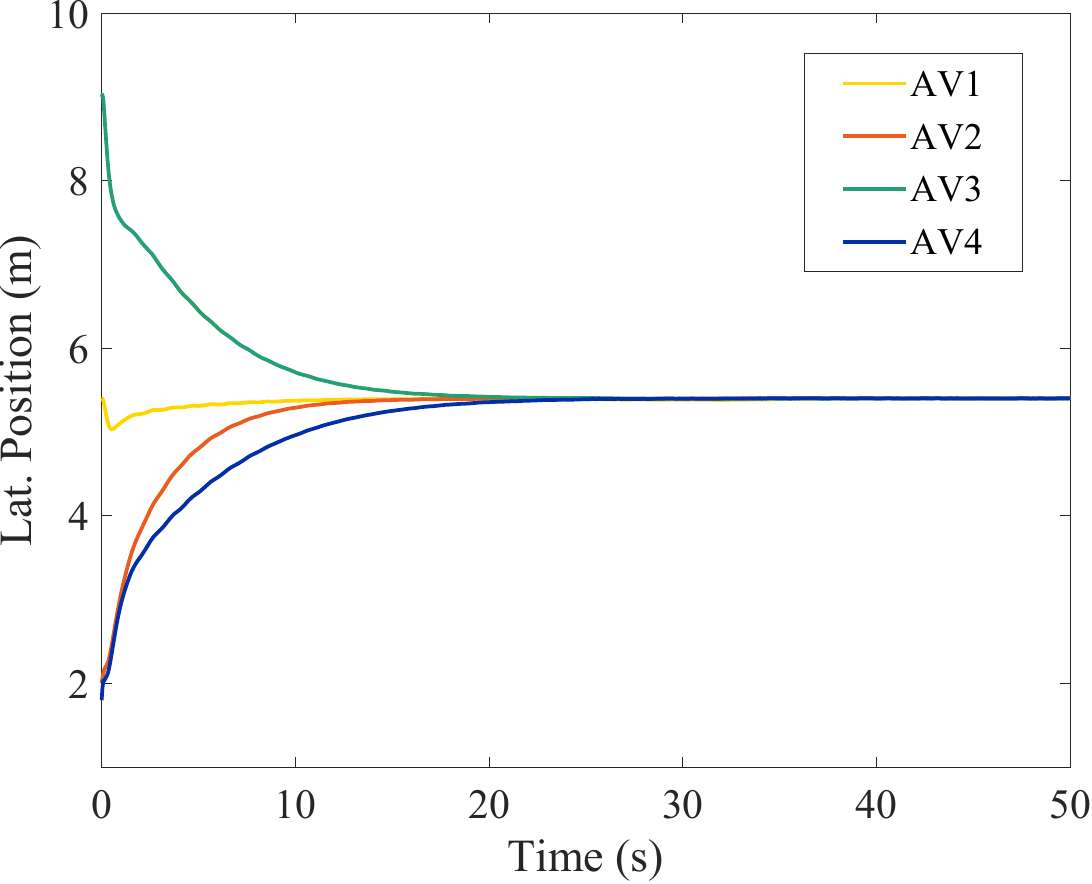}}
            \subfloat[ Square formation]{\includegraphics[width=0.33\linewidth]{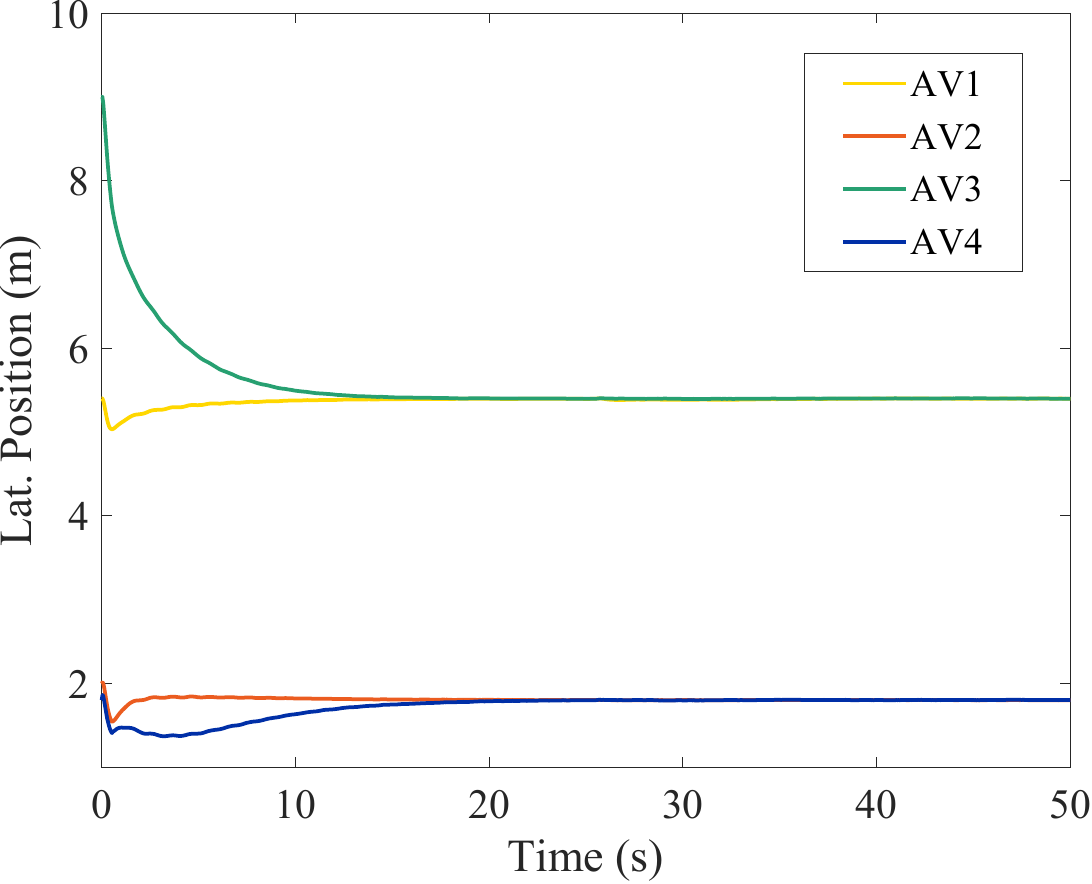}}
            \subfloat[Linear-queue formation]{\includegraphics[width=0.33\linewidth]{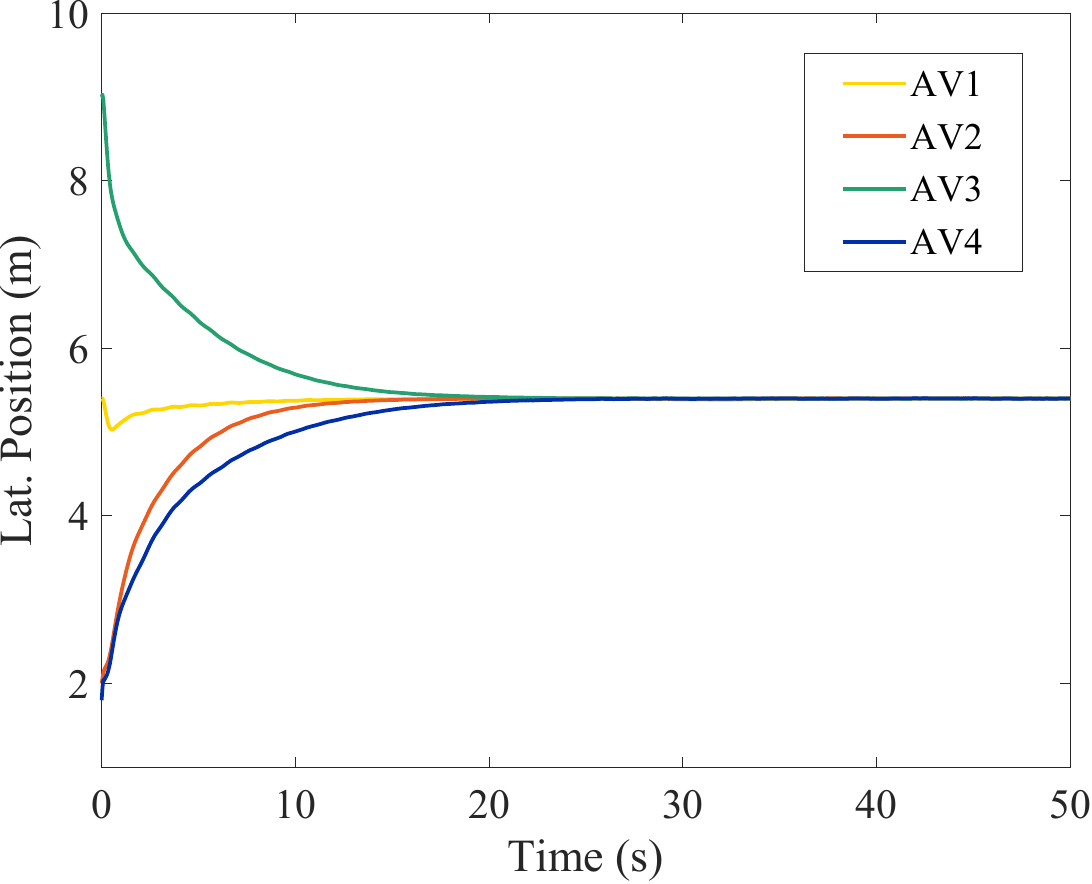}}
            \caption{Longitudinal and lateral tracking performance of the vehicle formation control in three-fold scenarios.}
            \label{tracking}
        \end{figure*}
	
    Taking the derivative of $V$ and based on the Young's inequality, we obtain
	\begin{equation}
		\label{equ18}
		\begin{aligned}
			\dot{V}&\leq-\beta V+\mu,\\
			\beta&=\min\{\phi_{\min}(K_{i,1}),\phi_{\min}(K_{i,2}),\phi_{\min}(\Upsilon_i\Delta_i), \\
			&\quad\Xi_{i,j}\phi_{\min}(O_{i,j})-2\phi_{\max}(O_{i,j}), \\
			&\quad(\phi_{\min}(Q)-3||P||^2)\phi_{\min}(P^{-1})\},\\
			\mu&=\sum_{i=1}^n[(\lambda_i^{\max})^2+(\varpi_i^{\max})^2+\frac{1}{2}\sum_{j=1}^2\Xi_{i,j}(\overline{W}_{i,j}^*)^2 \\
			&\quad+\varepsilon^*(\varepsilon_{i,1}+\varepsilon_{i,2})+\frac{1}{2}(\sigma_i-\sigma_i^0)^T\Upsilon_i(\sigma_i-\sigma_i^0)].
		\end{aligned}
	\end{equation}

	Then, it satisfies
	\begin{equation}
		\label{equ19}
		\begin{aligned}
			0\leq V(t)\leq\frac{\mu}{\beta}+(V(0)-\frac{\mu}{\beta})e^{-\beta t}.
		\end{aligned}
	\end{equation}
	
	Therefore, according to the properties of positive definite matrices, it is corroborated that the observer error $e$, tracking errors $z_{i,1}$ and $z_{i,2}$, parameter estimation errors $\widetilde{W}_{i,j}$ and $\widetilde{\sigma}_i$ are all bounded, i.e., $||e||^2\leq\frac{2V}{\phi_{\min}(P)}$, $||z_{i,1}||^2\leq2V$, $||z_{i,2}||^2\leq2V$, $||\widetilde{W}_{i,j}||^2\leq\frac{2V}{\phi_{\min}(O_{i,j}^{-1})}$ and $||\widetilde{\sigma}_{i,j}||^2\leq\frac{2V}{\phi_{\min}(\Delta_{i,j}^{-1})}$.

	Now we show that there exists a $t^*>0$  such that $\forall k\in Z^+$, $\{t_{k+1}-t_k\}\geq t^*$. To this end, by recalling $e_i(t)=w_i(t)-u_i(t)$, during the period $t_i^k\leq t_i<t_i^{k+1}$, we obtain
	\begin{equation}
		\label{equ19}
		\begin{aligned}
			\frac{d}{dt}||e_i||&=\frac{d}{dt}(e_i^T\times e_i)^{\frac{1}{2}}\\
			&=[{\rm sgn}(e_{i,1}),{\rm sgn}(e_{i,2})]^T	\dot{e}_i\leq||\dot{w}_i||,
		\end{aligned}
	\end{equation}
	where $e_{i,j}$ is the $j$th component of vector $e_i$, $j=1,2$.
	
	From ($\ref{equ8}$), we know that $w_i$ is continuous and bounded, and through the standard analysis in \cite{mu}, all closed-loop signals are bounded. Therefore, there must exist a positive unknown constant $\Psi_i\in \mathbb{R}$. Noting that $e_i(t_i^k)=0$ and $\lim_{t\rightarrow t_i^k}e_i(t)=\overline{\varsigma}_i$, we obtain that the event-triggered lower bound of inter-execution intervals $t_i^*$ must satisfy $t_i^*\geq(\overline{\varsigma}_i/\Psi_i)$, namely, the Zeno-behavior is successfully avoided.
    \end{proof}

    \begin{thm}
    The closed-loop system $(\ref{equ1})$, incorporating the continuous-in-time controller ($\ref{equ7}$) and the updated controller ($\ref{equ10}$) under the relative-threshold strategy ($\ref{equ11}$), ensures bounded tracking errors and signals for the AVs formation. Furthermore, the Zeno behavior is excluded.
    \end{thm}
    \begin{proof}
       Consider about the relative-threshold event-triggered strategy, from (\ref{equ11}), with time-varying parameters $\pi_1(t)$ and $\pi_2(t)$ satisfying $|\pi_1(t)|\leq1$ and $|\pi_2(t)|\leq1$, we have
	\begin{align}
		w_i(t)=(1+\pi_1(t)\zeta_i)\mu_i(t)+\pi_2(t)\xi_i. \label{53}
	\end{align}

    By applying \emph{Lemma 1}, we get
    \begin{equation}
		\label{equ_relative}
		\begin{aligned}
			z_{i,2}\frac{w_i(t)-\pi_2(t)\xi_i }{1+\pi_1(t)\zeta_i}\leq 2\varepsilon^*(\varepsilon_{i,1}+\varepsilon_{i,2})
		\end{aligned}.
	\end{equation}
	
	Following the same analysis in the proof of Theorem 1, we get that all the closed-loop signals are globally bounded.
	
	Finally, we are at the step to decide the lower bound of the inter-execution time $t^*$. Still following the analysis in Theorem 1 we have $t^*\geq(\zeta_i|u_i(t)|+\xi_i)/\Psi_i$, which is lower bounded. This means that the Zeno-behavior is successfully avoided.
    \end{proof}
	
        \begin{thm}
        The closed-loop system $(\ref{equ1})$, incorporating the continuous-in-time controller ($\ref{equ7}$) and the updated controller ($\ref{equ8}$)($\ref{equ10}$) under switched-threshold strategy ($\ref{equ14}$), including both the fixed-threshold strategy and the relative-threshold strategy, ensures bounded tracking errors and signals for the AVs formation. Furthermore, the Zeno behavior is excluded.
        \end{thm}

	\begin{proof}
	    Following the same analysis in the proof of Theorem 1 and Theorem 2, we get that all the closed-loop signals are globally bounded. Since the switched-threshold strategy adopts the same control laws as that in the initial two strategies, it is easy to get the time interval satisfy $t^*>\max\{(\overline{\varsigma}_i/\Psi_i,(\zeta_i|u_i(t)|+\xi_i)/\Psi_i\}$. This means that the Zeno-behavior is successfully avoided.
	\end{proof}

	\section{Illustrative Example}\label{sec5}
	In this section, an illustrative example is provided to validate the AVs formation control algorithm. Section V-A contains the  specification of all designed parameters utilized in the control algorithm. Section V-B demonstrates the tracking performance and the effects of the multi-threshold event-triggered mechanism in the AVs formation control algorithm. And we further analyze the safety and mobility performance of the controller in Section V-C.

    \begin{figure*}
            \centering
            \includegraphics[width=0.99\textwidth]{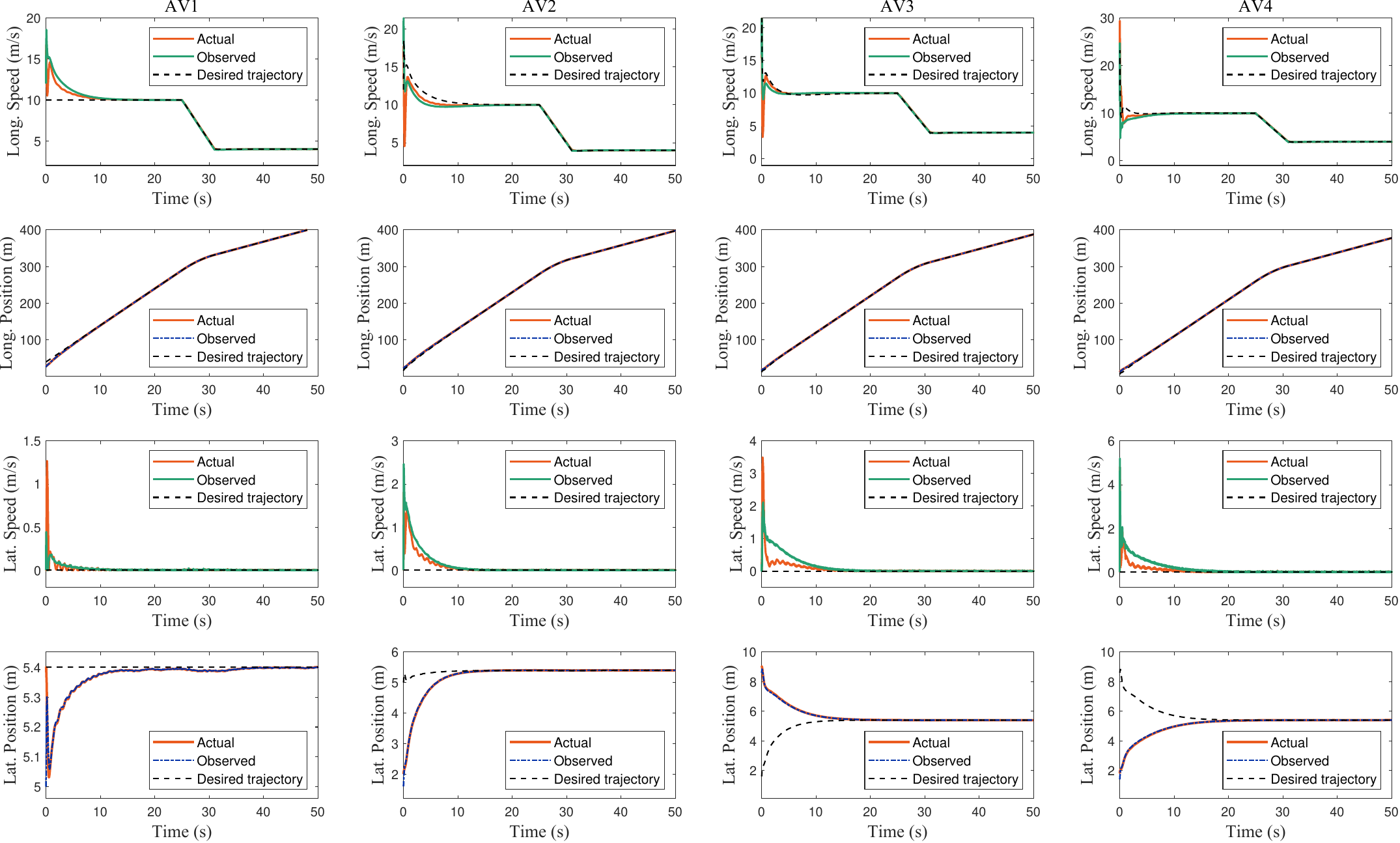}
            \caption{Comparative analysis of actual-observed-desired longitudinal/lateral positions and speeds for AVs in the linear formation under relative-threshold strategy. Each column provides the profiles of positions and speeds of one vehicle.}
            \label{speed&position}
          \end{figure*}

         \begin{figure}
            \centering
            \includegraphics[width=0.45\textwidth, height=0.27\textheight]{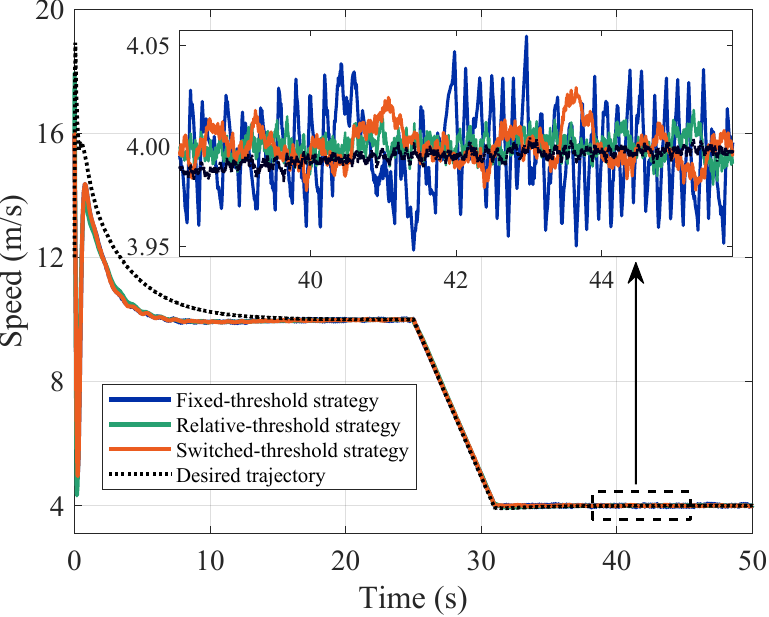}
            \caption{Comparison of longitudinal tracking speed of AV-$2$ under different event-triggered strategies in the linear formation.}
            \label{comparison}
        \end{figure}

        \begin{figure*}
            \centering
            \subfloat[ Fixed-threshold strategy]{\includegraphics[width=0.333\linewidth]{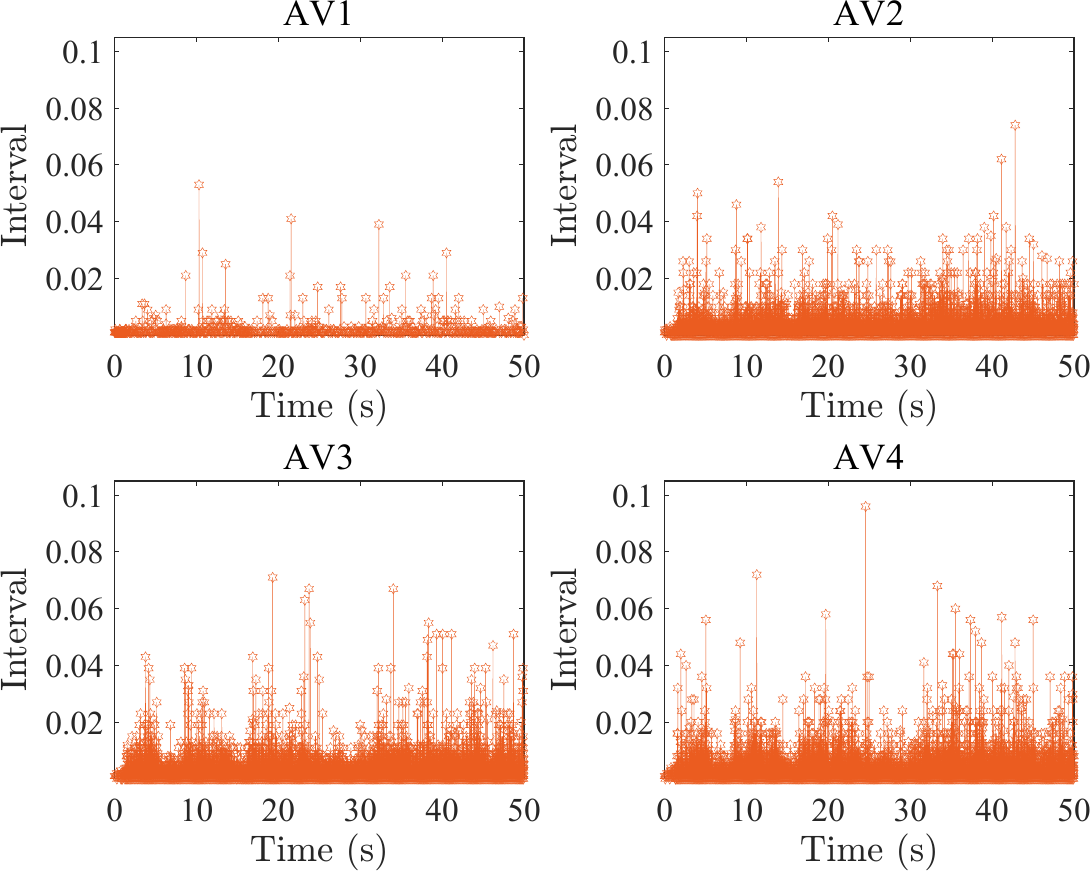}}
            \subfloat[Relative-threshold strategy]{\includegraphics[width=0.333\linewidth]{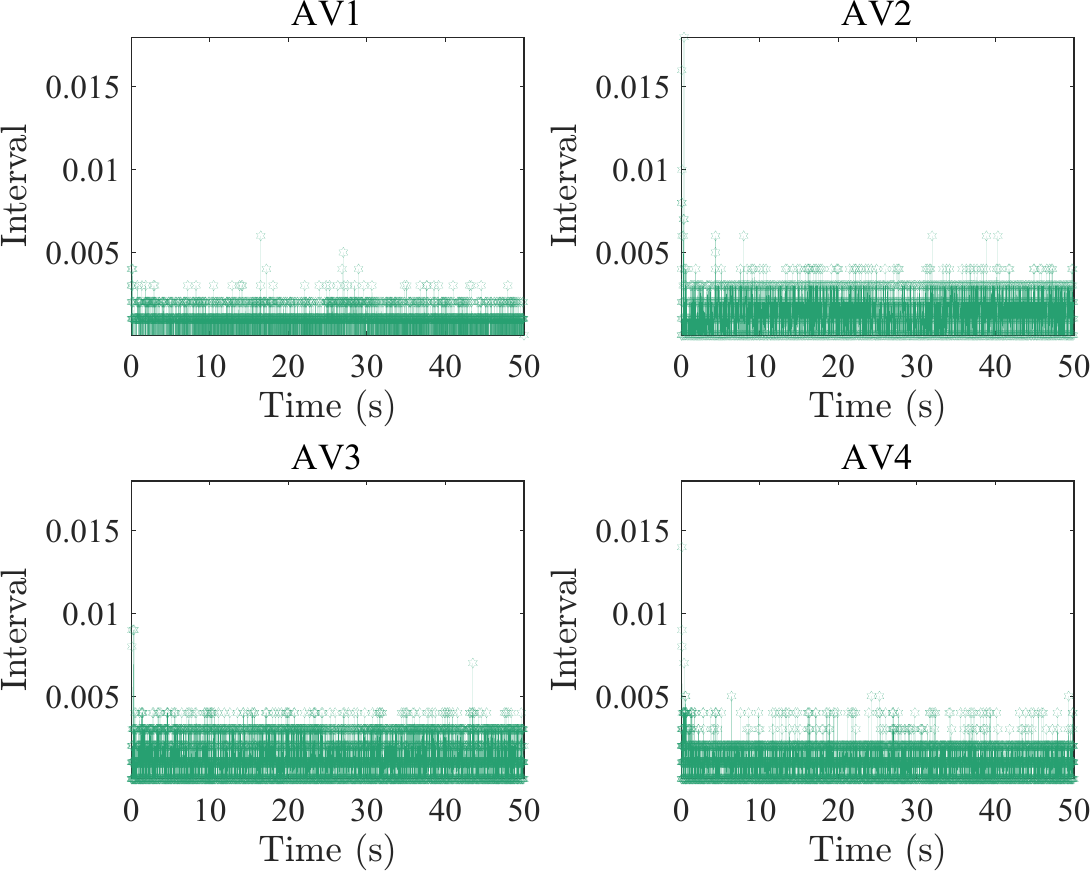}}
            \subfloat[Switched-threshold strategy]{\includegraphics[width=0.333\linewidth]{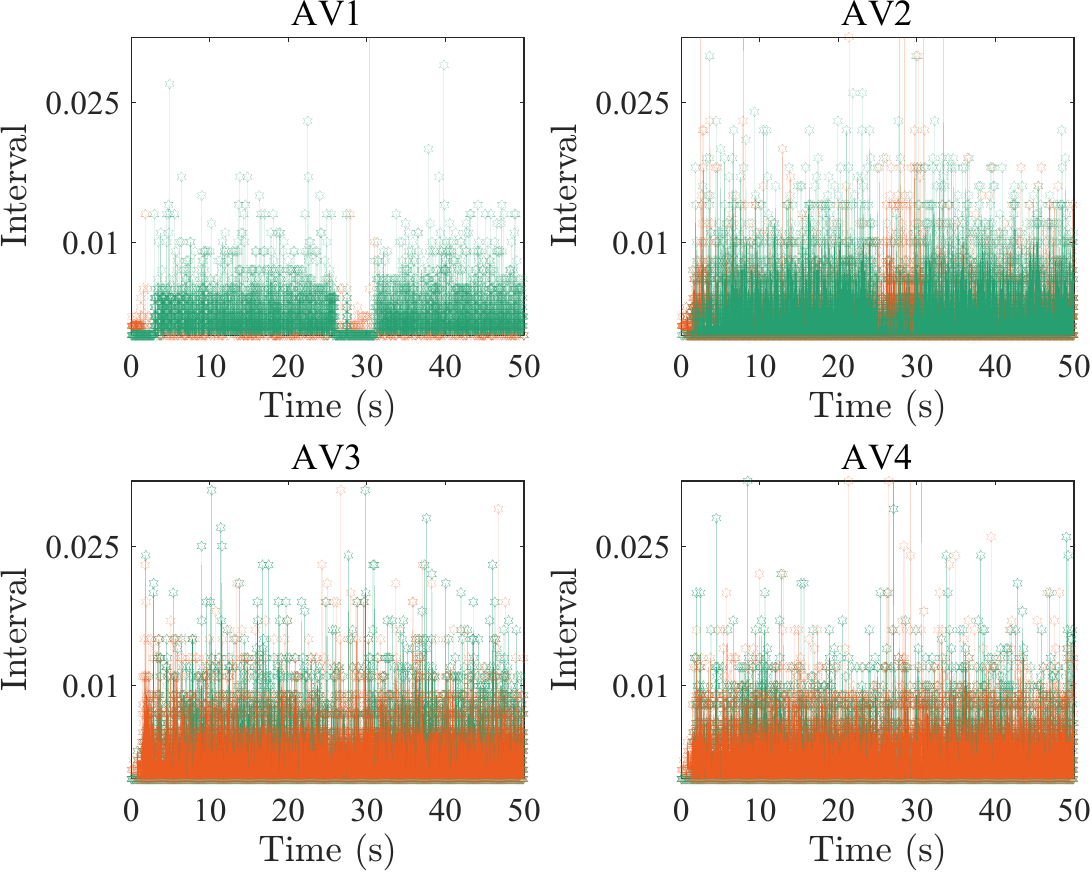}}
            \caption{The longitudinal update time interval of control laws for the AVs in linear formation.}
            \label{interval}
        \end{figure*}

        \begin{figure}%controller
		\centering
		\includegraphics[width=0.48\textwidth, height=0.23\textheight]{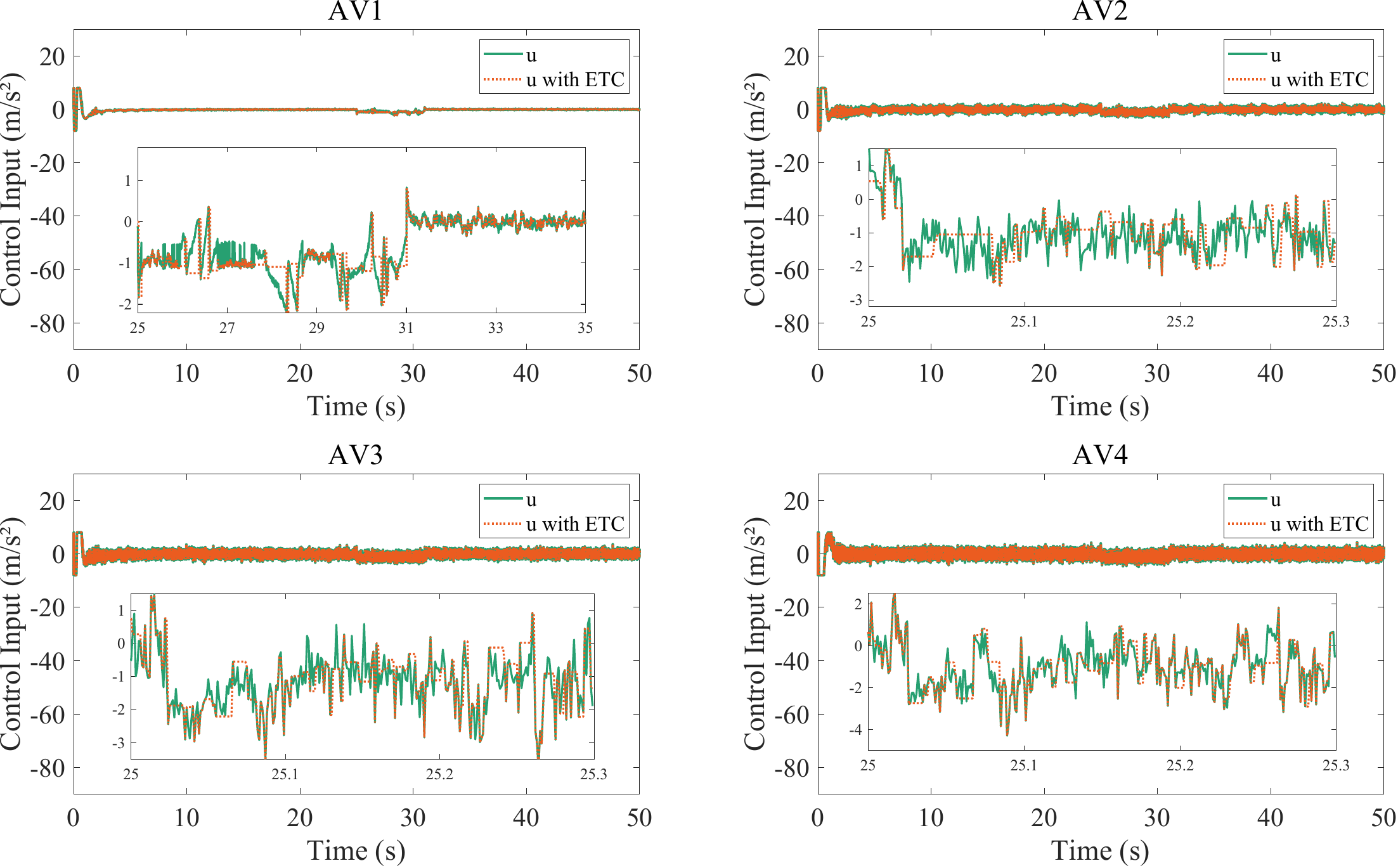}
		\caption{Performance of control inputs for AVs in linear formation under fixed-threshold strategy.}
		\label{controller}
	\end{figure}
    
    \begin{figure*}
            \centering
            \subfloat[Linear formation]{\includegraphics[width=0.333\linewidth]{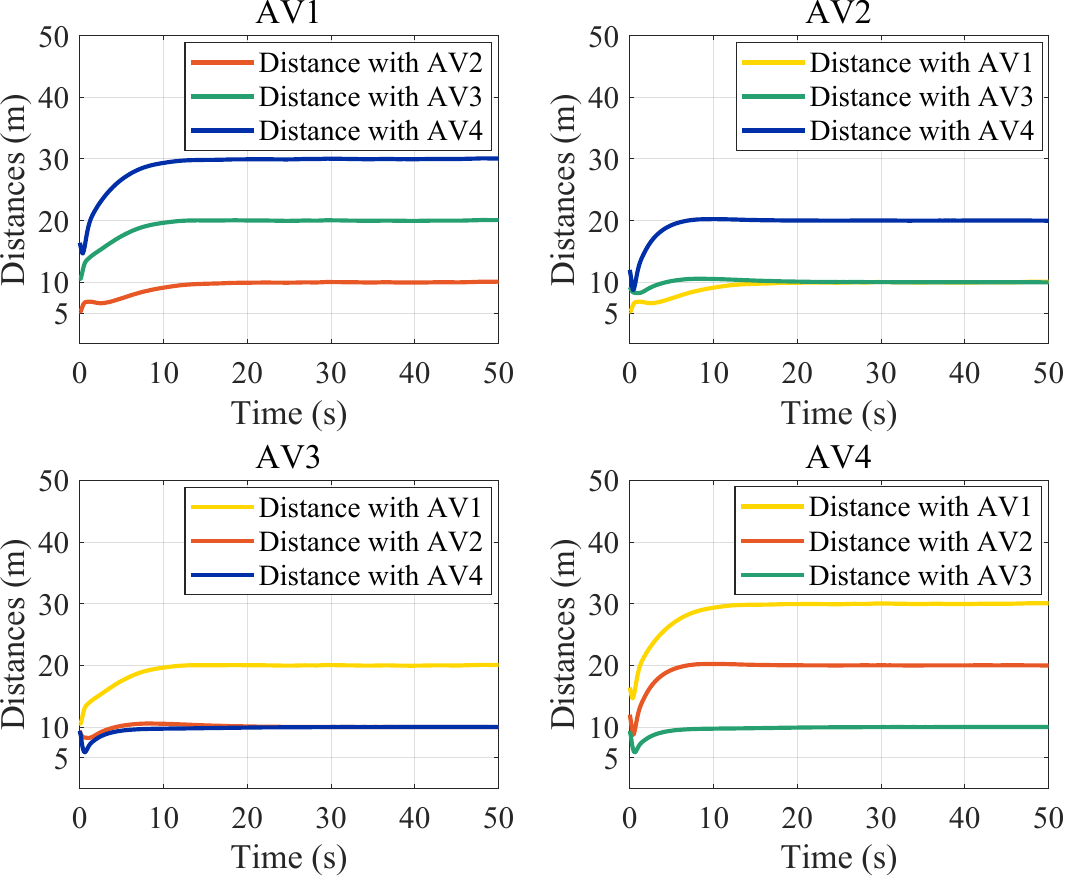}}
            \subfloat[Square formation]{\includegraphics[width=0.333\linewidth]{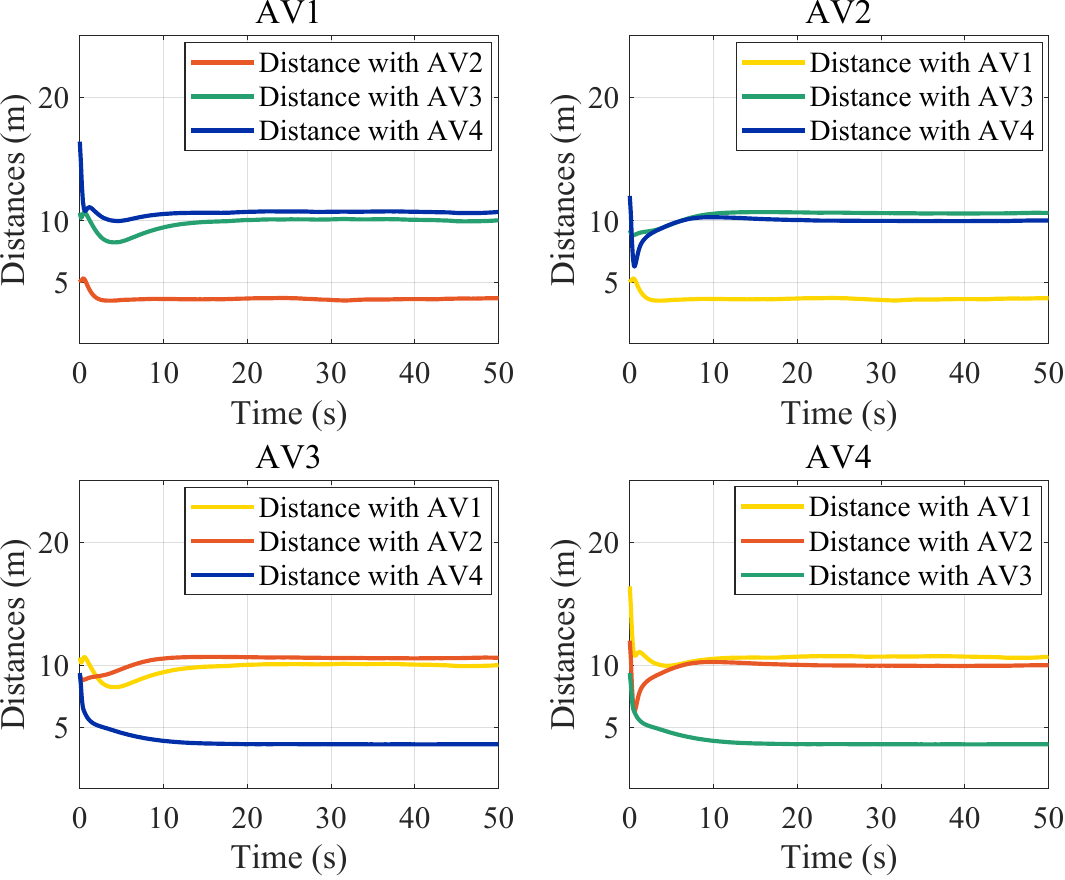}}
            \subfloat[Linear-queue formation]{\includegraphics[width=0.333\linewidth]{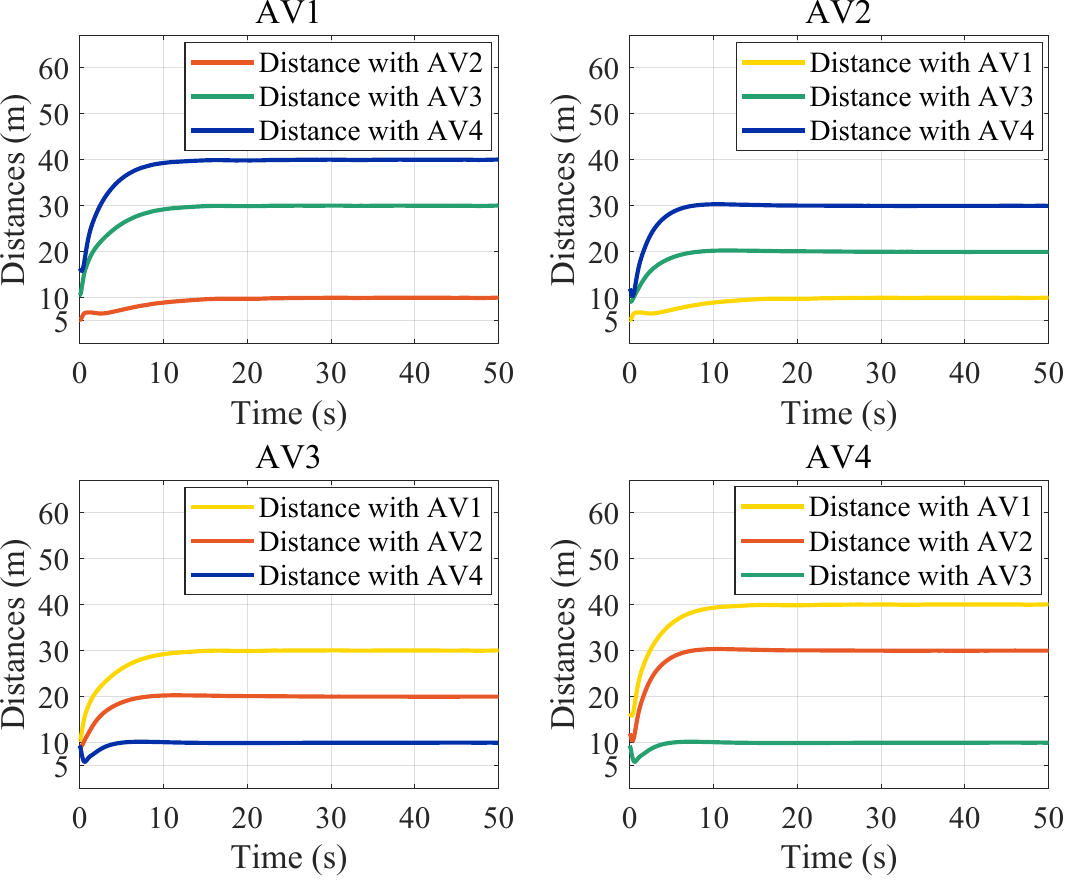}}
            \caption{The safe distances between each AV and the other AVs in different traffic scenarios.}
            \label{distance}
        \end{figure*}

        \begin{figure}
            \centering
            \subfloat{\includegraphics[width=0.48\linewidth]{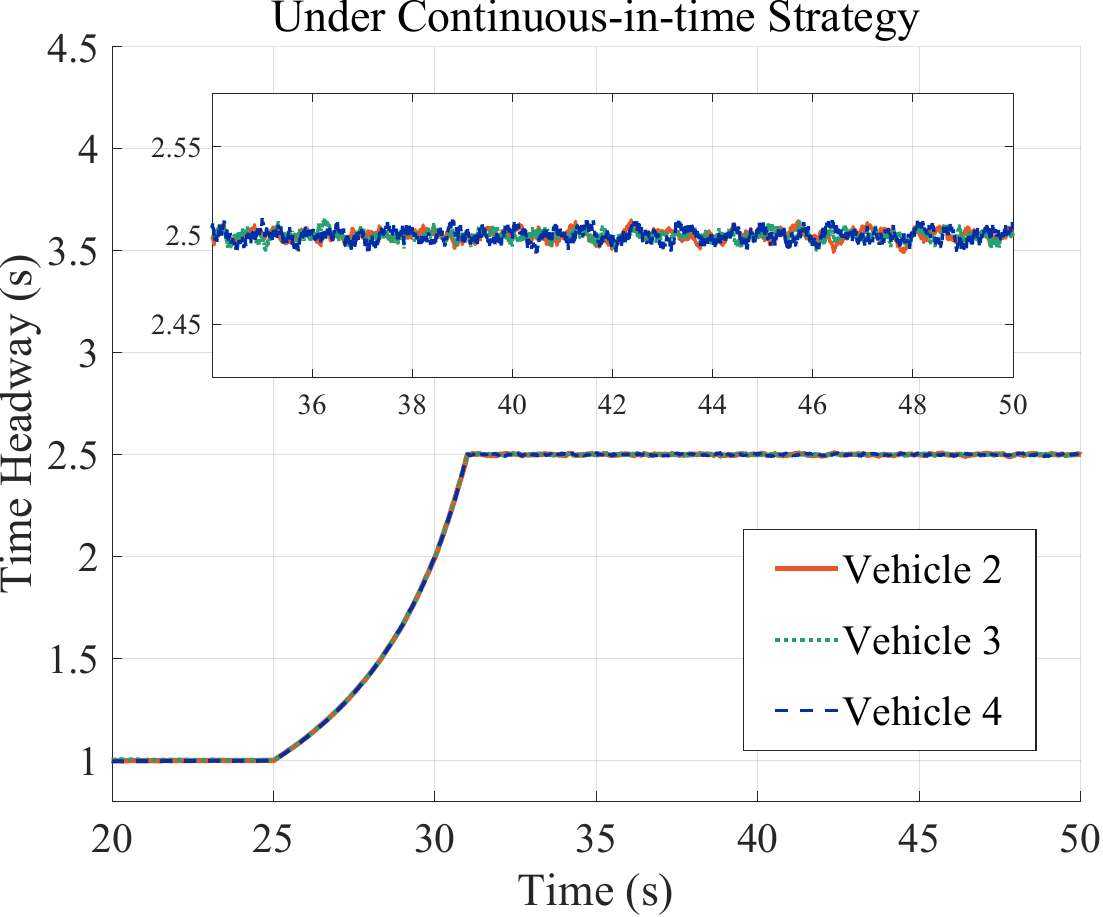}}
            \subfloat{\includegraphics[width=0.48\linewidth]{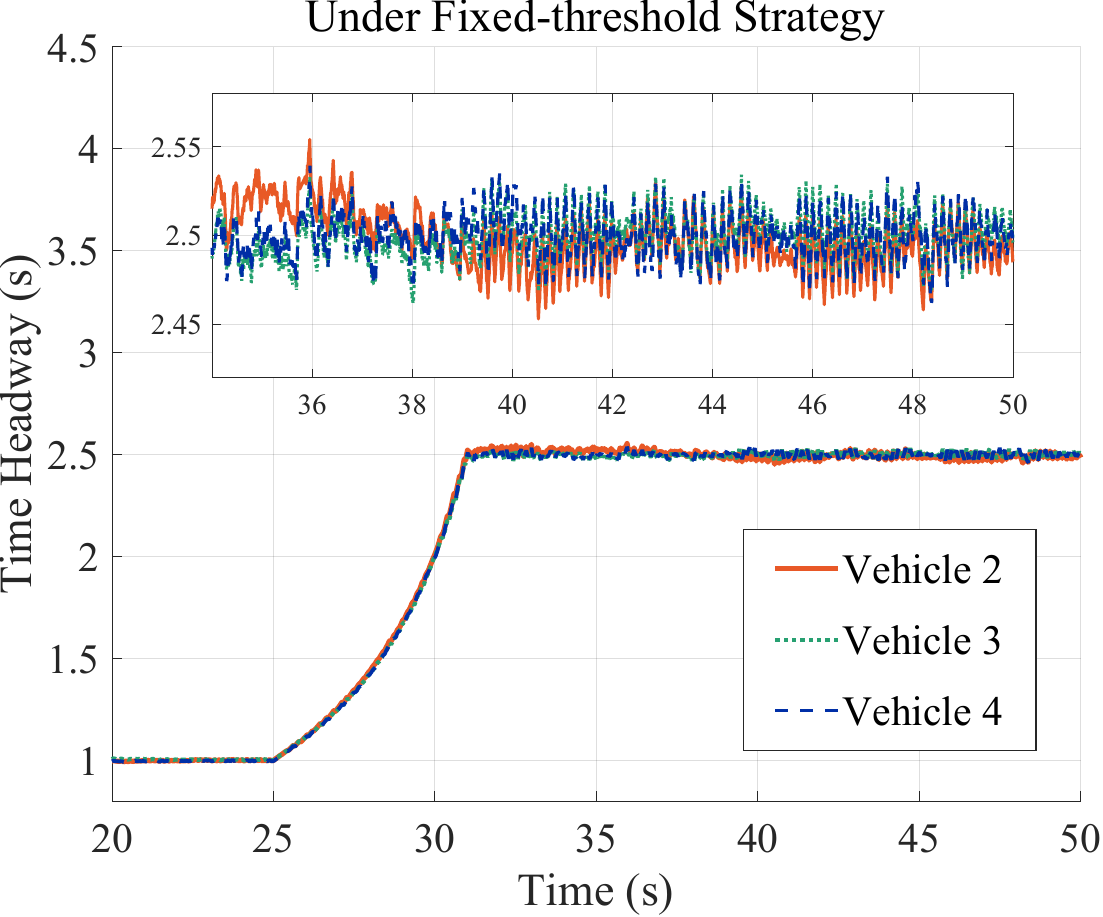}}\\
            \subfloat{\includegraphics[width=0.48\linewidth]{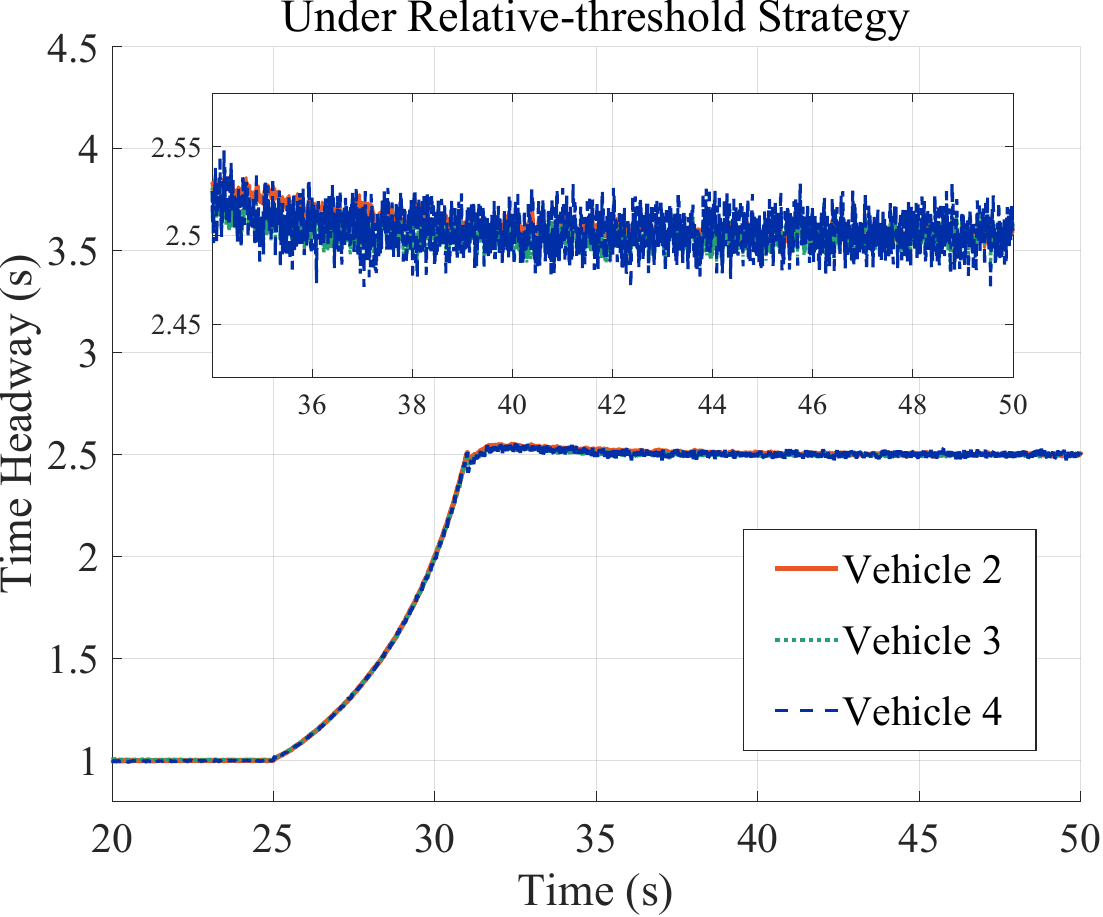}}
            \subfloat{\includegraphics[width=0.48\linewidth]{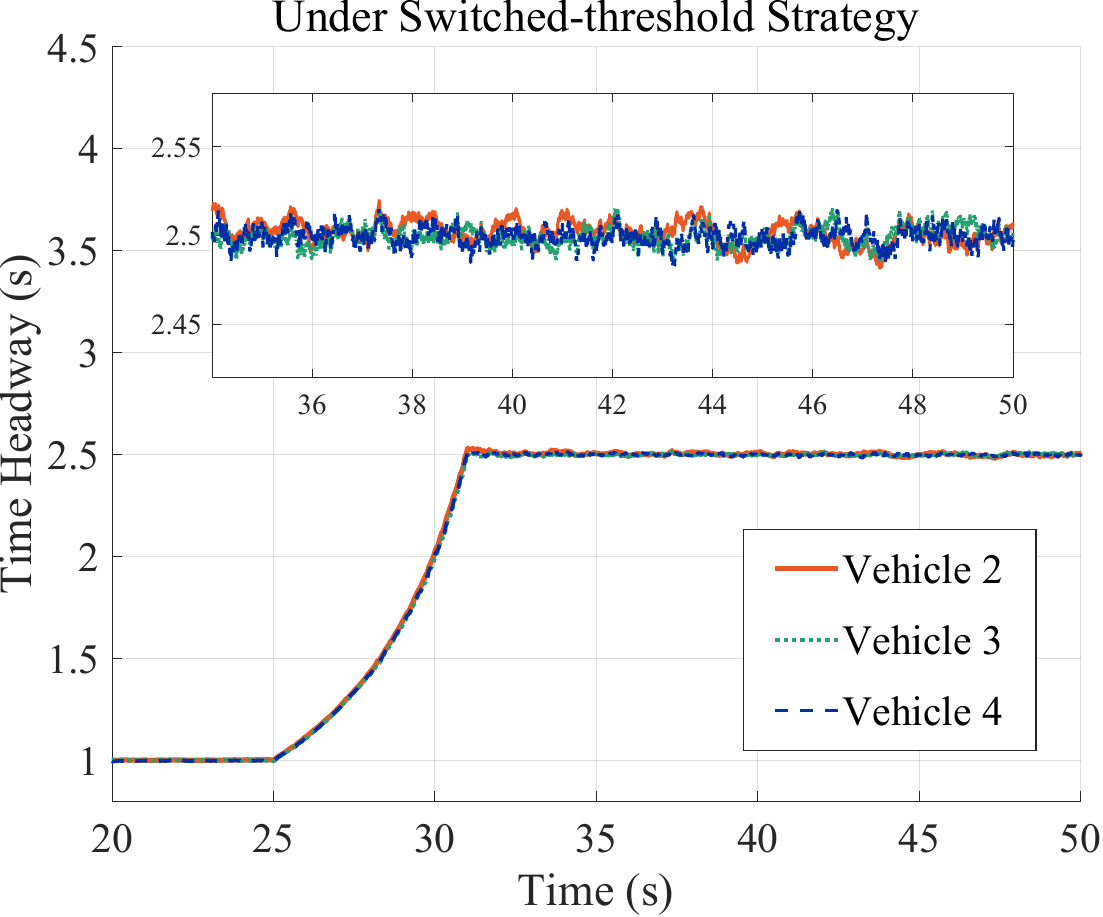}}
            \caption{The time headway under different control strategies in the linear formation.}
            \label{headway}
        \end{figure}
        
	\subsection{Parameter setting}
	We consider the formation of four vehicles. For the leading vehicle V-$1$, we adopt the expected longitudinal speed $v_1^{rx}$ and lateral speed $v_1^{ry}$ trajectories used in \cite{speed} as
	\begin{equation}
		\label{equ23}
		v_1^{rx} = 
		\begin{cases}
			10\ \mathrm{m/s}, & \text{ } 0\ \mathrm{s} \leq t < 25\ \mathrm{s} \\
			-(t-25) + 10\ \mathrm{m/s}, & 25\ \mathrm{s} \leq t < 31\ \mathrm{s} \\
			4\ \mathrm{m/s}, & 31\ \mathrm{s} \leq t \leq 50\ \mathrm{s}
		\end{cases} ,
	\end{equation}
	\begin{equation}
		\label{equ24}
		v_1^{ry} = 0\ \mathrm{m/s}, \quad 0\ \mathrm{s} \leq t < 50\ \mathrm{s}
	\end{equation}
	
	For simplicity, according to \cite{mu}, we assume the uniform resistance and external disturbances across all AVs within the formation, and the unknown longitudinal resistances $f_i^x$, lateral resistances $f_i^y$ and external disturbances $P_i$ are represented as follows
	\begin{equation}
		\label{equ25}
		\begin{bmatrix}
			f_i^{x} \\ 
			f_i^{y}
		\end{bmatrix} 
		= 
		\begin{bmatrix}
			0.5H_1 \times H_2 \times H_3 \times (v_i^{x})^2 \\ 
			0.5H_1 \times H_2 \times H_3 \times (v_i^{y})^2
		\end{bmatrix},
	\end{equation}
	\begin{equation}
		\label{equ26}
		P_i = 
		\begin{bmatrix}
			0.3 \times \sin(2\pi t) \times e^{-t/5} \\ 
			0.3 \times \sin(2\pi t) \times e^{-t/5}
		\end{bmatrix}\ \mathrm{m/s^2},
	\end{equation}
	where $H_1=1.206\ \mathrm{kg/m^3}$ is the air density, $H_2 = 5.58\ \mathrm{m^2}$ denotes the cross-sectional area of AVs, and $H_3 = 0.3$ represents the dimensionless drag coefficient. The initial position and speed of all vehicles and observers are presented in Table \ref{table2}. The expected inter-vehicle distances $l_i=[l_i^x,l_i^y]^T$ are presented in Table \ref{table4}.

    The parameters used in the control process are shown as follows: the mass of vehicles are $m_1 = 1760$ kg, $m_2 = 1920$ kg, $m_3 = 1660$ kg, $m_4 = 1890$ kg; the observer output injection gains are $C_{i,1} = {\rm diag}(5,5)$, $C_{i,2} = {\rm diag}(50,50)$; the definite matrices in controller design are $K_{i,1}={\rm diag}(0.5,0.5)$, $K_{i,2}={\rm diag}(20,20)$; the number of hidden layers in neural networks is $l = 5$; the fixed-threshold strategy takes $\varsigma_i = 2$, $\overline{\varsigma}_i = 2.5$, $\varepsilon_{i,1}=\varepsilon_{i,2}=0.5$; the relative-threshold strategy takes $\zeta_i = 0.9$, $\xi_i = 0.1$, $\overline{\xi} = 2$; the switched-threshold strategy takes $S = 0.55$; other matrices are $\Upsilon_i = {\rm diag}(2,2)$, $\Delta_i = {\rm diag}(0.2,0.2)$.

	\begin{table}[!t]
		\caption{Initial status of vehicles and observers}
		\label{table2}
		\centering
		\footnotesize
		\begin{tabular}{ccccc}
			\toprule
			\textbf{ } & 
			\multicolumn{2}{c}{\textbf{Initial Position (m)}} & 
			\multicolumn{2}{c}{\textbf{Initial Speed (m/s)}} \\
			\cmidrule(lr){2-3} \cmidrule(lr){4-5} 
			$No.$ & $x_i(0)$ & $\hat{x}_i(0)$ & $v_i(0)$ & $\hat{v}_i(0)$ \\
			\midrule
			AV1 & (28, 5.4) & (26, 5.0) & (14, 0) & (12, 0)  \\
			AV2 & (24, 2.0) & (22, 1.6) & (16, 0) & (18, 0)  \\
			AV3 & (18, 9.0) & (16, 8.6) & (16, 0) & (16, 0)  \\
			AV4 & (12, 1.8) & (14, 1.4) & (17, 0) & (14, 0)  \\
			\bottomrule
		\end{tabular}
	\end{table}

	% \begin{table}[!t]
	% 	\caption{Simulation parameters}
	% 	\label{table3}
	% 	\centering
	% 	\footnotesize
	% 	\begin{tabular}{cccccl}
	% 		\toprule
	% 		\textbf{Parameter} & \textbf{Value} & \textbf{Parameter} & \textbf{Value} & \textbf{Parameter} & \textbf{Value}\\
	% 		\midrule
	% 		$m_1$ & 1760 kg & $C_{i,1}$ & 5  & $K_{i,1}$ & diag(0.5,0.5)  \\
	% 		$m_2$ & 1920 kg & $C_{i,2}$ & 50 & $K_{i,2}$ & diag(20,20)  \\
	% 		$m_3$ & 1660 kg & $\varepsilon_{i,1}$ & 0.5 & $\Delta_i$ & diag(0.2,0.2) \\
	% 		$m_4$ & 1890 kg & $\varepsilon_{i,2}$ & 0.5 & $\Upsilon_i$ & diag(2,2) \\
	% 		\bottomrule
	% 	\end{tabular}
	% \end{table}

	\begin{table}[!t]
		\caption{The expected inter-vehicle distances}
		\label{table4}
		\centering
		\footnotesize
		\begin{tabular}{cccc}
			\toprule
			\textbf{No.} & \textbf{Linear} & \textbf{Square} & \textbf{Linear-queue} \\
			\midrule
			AV1 & (0,0) & (0,0) & (0,0) \\			
                AV2 & (10,0) & (0,3.6) & (10,0)   \\
			AV3 & (10,0) & (10,-3.6) & (20,0)  \\
			AV4 & (10,0) & (0,3.6) & (10,0)  \\
			\bottomrule
		\end{tabular}
	\end{table}

	\subsection{Analysis on AVs formation control algorithm}
        The three-fold formation scenarios shown in Fig. \ref{fig1} are tested. Fig. \ref{tracking} presents the longitudinal and lateral position trajectories of all AVs in different formations. It is evident that the longitudinal trajectory changes are completed rapidly within 4 s, while the lateral changes require 20 s to complete. Moreover, it is noteworthy that only the first AV receives precise trajectories for expected position and speed.

        Fig. \ref{speed&position} compares observed and actual longitudinal/lateral positions and speeds of four AVs in the linear formation under the relative-threshold strategy. Significant deviations between observed and actual values are evident at the initial time, primarily due to estimation errors in the sampling-based observer. Initial inaccuracies in the parameter estimation lead to significant deviations. However, observed values converge to actual positions and speeds over time. Notably, the sampling-based observer effectively tracks both positions and speeds despite relying solely on imprecise sampled position data.

        Table \ref{table5} and Figs. \ref{comparison}-\ref{controller}  illustrate the performance of the proposed multi-threshold event-triggered strategies. Table \ref{table5} presents the number of triggers corresponding to three different threshold event-triggered control strategies. Notably, the data in the switched-threshold strategy includes the total number of switched triggers as well as the number of triggers for each of the two switched strategies. Fig. \ref{comparison} illustrates variations in the longitudinal speed of AV-$2$ within the linear formation scenario under different event-triggered strategies. It is evident that the initial frequencies of speed variation are similar across the three strategies. However, once the tracking speed stabilizes, the relative-threshold strategy exhibits the highest level of control precision, while the fixed-threshold strategy shows the lowest. The switched-threshold strategy offers a moderate level of control precision which can be theoretically adjusted by modifying the switching boundary, $S$. As $S$ increases, the proportion of the fixed-threshold strategy in the triggering policy will increase, with the control objective shifting more towards reducing the number of triggers. Conversely, the relative-threshold strategy will gain greater prominence in the triggering policy, with the control objective becoming more focused on ensuring control performance. Furthermore, the updated time intervals of the control laws for each vehicle are displayed in Fig. \ref{interval}. We also provide a comparison between the continuous controller using the backstepping method and the event-triggered controller for AVs in close formation under the fixed-threshold strategy in Fig. \ref{controller}. By utilizing the event-triggered strategies, the traffic system avoids the necessity of updating the control input at every time step, thereby significantly reducing the computational burden.

        \begin{table}[!t]
		\caption{Trigger counts across the multi-threshold event-triggered strategies}
		\label{table5}
		\centering
		\footnotesize
		\begin{tabular}{ccccc}
			\toprule
			\textbf{No.} & \textbf{Fixed} & \textbf{Relative} & \textbf{Switched} & \textbf{Continuous} \\
			\midrule
			AV1 & 1888 & 7111 & 7033 (1319+6144) & 50000  \\			
                AV2 & 15197 & 44711 & 24314 (9857+14457) & 50000  \\
			AV3 & 24101 & 45752 & 28827 (20719+8108) & 50000  \\
			AV4 & 29904 & 46164 & 33414 (27917+5497) & 50000  \\
			\bottomrule
		\end{tabular}
	\end{table}

 %    \begin{table}[!t]
	% 	\caption{Statistics of average time headway under different control strategies in the linear formation for AV-$i$ ($i=2,3,4$) during the time interval of 35 to 50 seconds.}
	% 	\label{table5}
	% 	\centering
	% 	\footnotesize
	% 	\begin{tabular}{ccccc}
	% 		\toprule
	% 		\textbf{No.} & \textbf{Continuous} & \textbf{Fixed} & \textbf{Relative} & \textbf{Switched}  \\
	% 		\midrule			
 %                AV2 & 2.5001 & 2.5168 & 2.5054 & 2.5026  \\
	% 		AV3 & 2.5000 & 2.4916 & 2.4988 & 2.4962  \\
	% 		AV4 & 2.4999 & 2.5032 & 2.5028 & 2.5008  \\
	% 		\bottomrule
	% 	\end{tabular}
	% \end{table}

        \subsection{Analysis of the traffic safety and mobility}
        In this paper, we analyze the safety and mobility performance of the proposed controller. For safety, we illustrate the distances between AVs during formation control processes as a safety measurement, considering their movements in both longitudinal and lateral directions, as shown in Fig. \ref{distance}. The results show that in scenarios 1 and 3, the minimum distance between AVs remains larger than 5 meters, while in scenario 2, it is constantly above 4 meters. These distances are consistent with safety standards outlined in \cite{yu2021automated}, indicating that the proposed formation control strategy effectively ensures the safety of road traffic performance. The dynamic adjustments in distances reflect the system's ability to maintain safe spacing while transitioning between formations, further validating its robustness in real-world applications.

        For the mobility analysis, we focus on the scenario 1 of linear formation, which has a clear mobility definition. According to \cite{ultra}, the mobility is measured by the time headway, defined by the time difference between consecutive arrival instants of two vehicles passing a certain detector site on the same lane. The time headway is considered a direct measure of road capacity. A short time headway increases road capacity and thus increases mobility, and vice versa. The time headway is be calculated as:
        \begin{equation}
		\label{equ25}
		\begin{aligned}
			\tau^*=\frac{h^*}{v^*}
		\end{aligned},
	\end{equation}
        where $\tau^*$ denotes the value of time headway, $h^*$ and $v^*$ represent the distance between the center of two vehicles and the speed of the subsequent vehicle respectively. 
        
        From Fig. \ref{tracking}, considering the linear formation, we find out that all AVs  fully transition into operating within the same lane at 20 s. Therefore, we only consider the changes in time headway for the proposed formation control strategies from 20 s to 50 s. 
        Fig. \ref{headway} presents the time headway of following AVs under different control strategies. Notably, the variations in the curves of Fig. \ref{headway} are attributed to changes in the formation speed. To elaborate, the time headway transitioned from 1 second to 2.5 seconds during the interval from 25 s to 31 s as a result of the formation vehicle speed reducing from 10 m/s to 4 m/s. This change reasonably corroborates the conclusion presented in \cite{ultra}.
        
        In Fig. \ref{headway}, we also magnify the section of the image where the time headway changes stabilize. Additionally, Table \ref{timeheadway} records the range of maximum and minimum time headway under different control strategies in the linear formation for AV-$i$ ($i=2,3,4$) during the time interval of 35 s to 50 s. In the case where the continuous-in-time strategy is used as the benchmark, the stabilized time headway variation range corresponding to the three ETC strategies is larger, but notably, the statistical results indicate that the switched-threshold ETC strategy exhibits the smallest variation, closely resembling the continuous-in-time control strategy. Overall, the implementation of ETC strategies results in a subtle impact on the range of time headway variation for each following vehicle during stable vehicular formation, with the switched-threshold strategy exerting the least influence.
        
 \begin{table}[!t]
		\caption{Statistics of the range of maximum and minimum time headway (s) under different control strategies in the linear formation for AV-$i$ ($i=2,3,4$) during the time interval of 35 to 50 seconds.}
		\label{timeheadway}
		\centering
		\footnotesize
		\begin{tabular}{ccccc}
			\toprule
			\textbf{No.} & \textbf{Continuous} & \textbf{Fixed} & \textbf{Relative} & \textbf{Switched}  \\
			\midrule			
                AV2 & 0.0184  & 0.0955 (5.1$\times$) & 0.0416 (2.2$\times$) & \textbf{0.0212} (\textbf{1.2$\times$})  \\
			AV3 & 0.0175 & 0.0808 (4.6$\times$) & 0.0368 (2.1$\times$) & \textbf{0.0294} (\textbf{1.7$\times$})  \\
			AV4 & 0.0233 & 0.0746 (3.2$\times$) & 0.0678 (2.9$\times$) & \textbf{0.0324} (\textbf{1.4$\times$})  \\
			\bottomrule
		\end{tabular}
	\end{table}
        
	\section{Conclusion}\label{sec6}

    In this paper, we proposed an observer-based event-triggered adaptive formation control framework for autonomous vehicles (AVs) with longitudinal and lateral motion uncertainties in traffic scenarios involving obstacles and cut-ins. A sampling-based observer was designed to estimate states with intermittent positional measurements, and a backstepping continuous-time boundary controller enabled accurate formation tracking. To enhance control efficiency, we introduced three event-triggered control (ETC) strategies i.e. fixed-threshold, relative-threshold, and switched-threshold, to regulate control signal updates. The relative-threshold strategy achieved the highest tracking precision due to more frequent triggering, while all three strategies ensured system safety. Notably, the switched-threshold strategy maintained vehicular formation mobility with minimal performance impact.

    The adaptive observed event-triggered formation control brings several promising future research directions. It is of interest to introduce the self-triggered control to predict the next triggering time, thereby avoiding the need for real-time monitoring of the controller's value when determining triggering conditions. Another extension is to consider the impact of more complex AV-human interactions on formation control in mixed traffic.

	%\appendix[PROOF OF THEOREM 1]

\flushend

% \vspace{-30pt}

\begin{IEEEbiography}[{\includegraphics[width=1in,height=1.25in,clip,keepaspectratio]{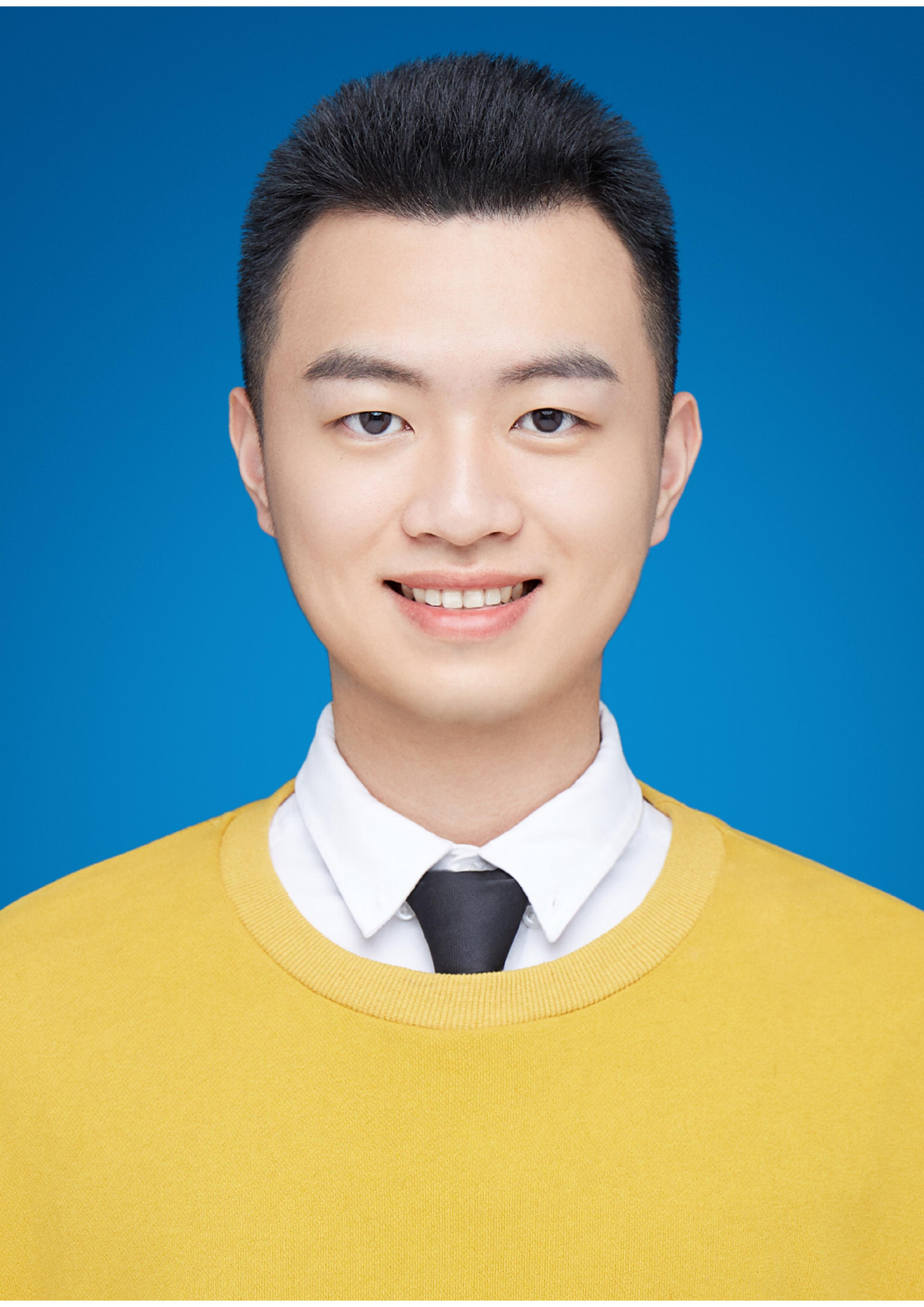}}]{Ziming Wang} (Graduate Student Member, IEEE) received the B.E. degree in Electronic Information Engineering from Southwest University in 2023. He is currently an MPhil student in the Trust of Robotics and Autonomous Systems at the Hong Kong University of Science and Technology (Guangzhou). He is interested in control theory, optimization, reinforcement learning, and their applications in intelligent transportation systems.
\end{IEEEbiography}

% \vspace{-30pt}

    \begin{IEEEbiography}[{\includegraphics[width=1in,height=1.25in,clip,keepaspectratio]{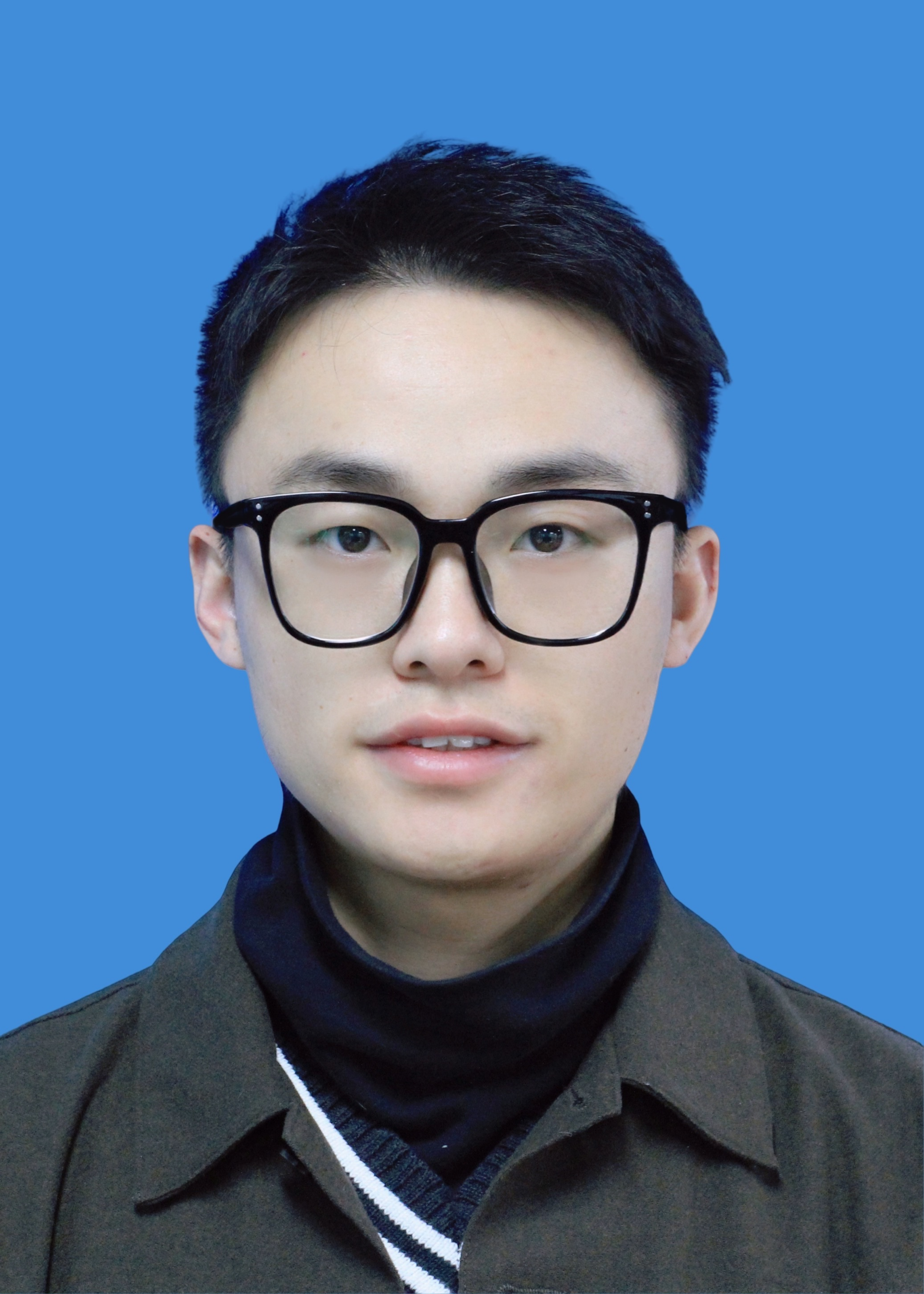}}]{Yihuai Zhang} (Graduate Student Member, IEEE) received his B.E. degree in Vehicle Engineering from Southwest University in 2019, and his M.S. degree in Vehicle Engineering from South China University of Technology in 2022. He is currently a Ph.D. student in the Thrust of Intelligent Transportation at the Hong Kong University of Science and Technology (Guangzhou). His doctoral research focuses on distributed parameter systems, learning and control for dynamical systems, and their applications in intelligent transportation systems.
\end{IEEEbiography}

% \vspace{-30pt}

    \begin{IEEEbiography}[{\includegraphics[width=1in,height=1.25in,clip,keepaspectratio]{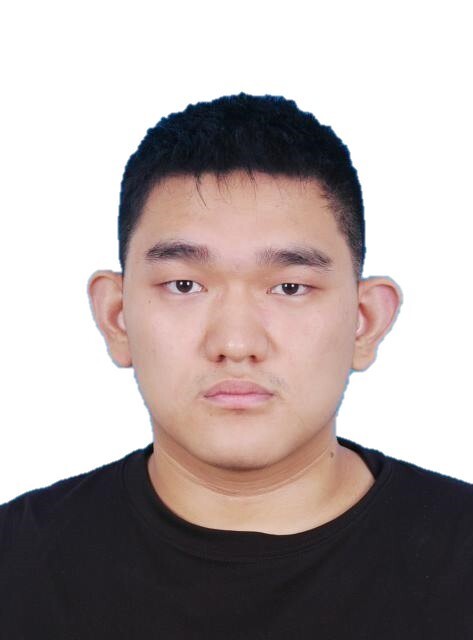}}]{Chenguang Zhao} (Graduate Student Member, IEEE) received the B.Sc. degree in electronic and information engineering and the M.Sc. degree in transportation information engineering from Beihang University in 2019 and 2022, respectively. He is currently pursuing his Ph.D. degree at the Hong Kong University of Science and Technology (Guangzhou). His research interests include CAV platoon control, traffic flow modeling and control.
\end{IEEEbiography}

% \vspace{-30pt}
    
    \begin{IEEEbiography}
[{\includegraphics[width=1in,height=1.25in,clip,keepaspectratio]{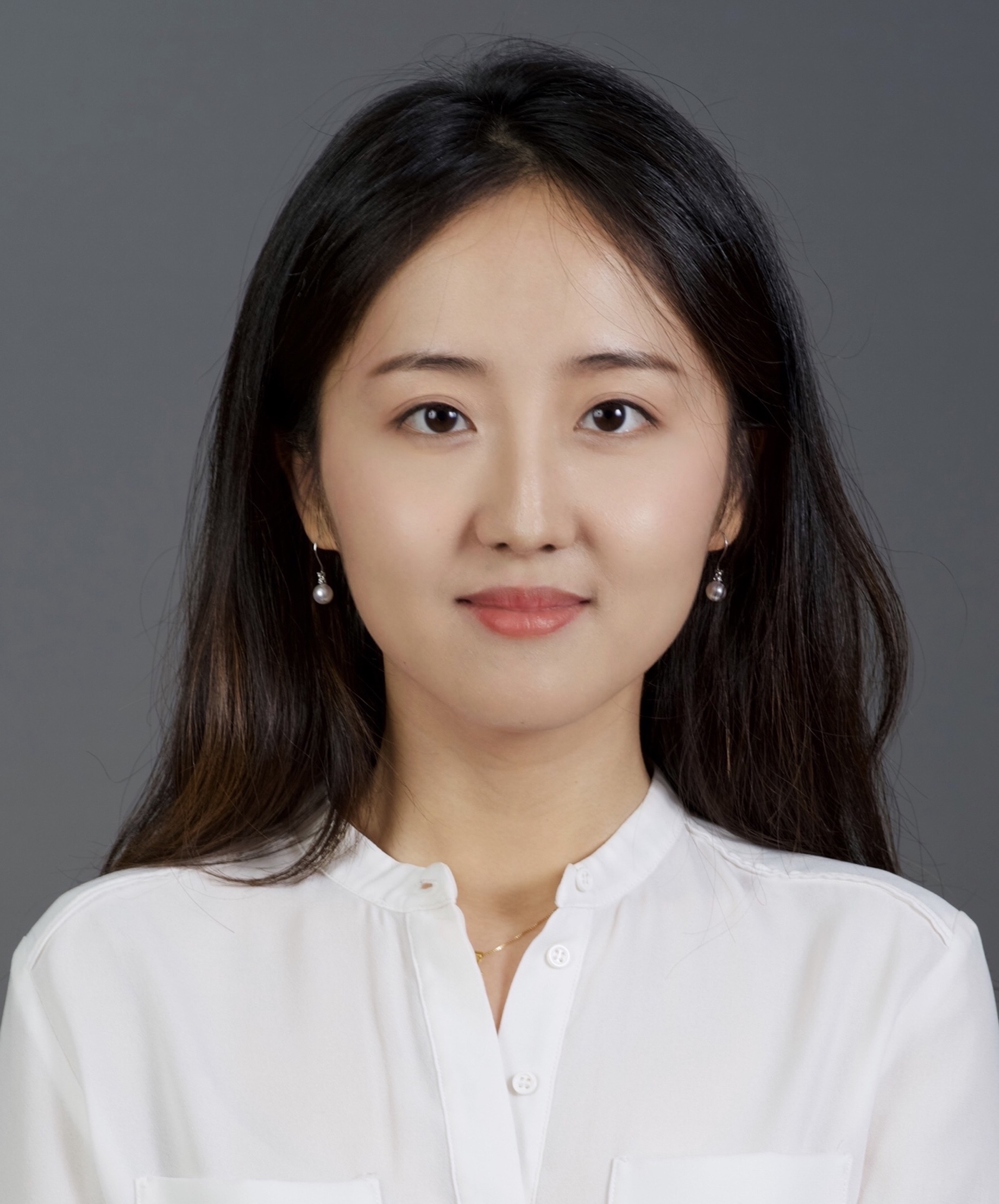}}]{Huan Yu} (Senior Member, IEEE) is an Assistant Professor in the Intelligent Transportation Thrust at the Hong Kong University of Science and Technology (Guangzhou). Yu received her B.Eng. degree in Aerospace Engineering from Northwestern Polytechnical University, and the M.Sc. and Ph.D. degrees in Aerospace Engineering from the Department of Mechanical and Aerospace Engineering, University of California, San Diego. She was a visiting scholar at University of California, Berkeley and Massachusetts Institute of Technology. She is broadly interested in control theory, optimization, and machine learning methodologies and their applications in intelligent transportation systems.
\end{IEEEbiography}
\end{document}